\documentclass[conference]{IEEEtran}

\usepackage{cite}

\usepackage{amsmath,amssymb,amsfonts}

\usepackage{algorithmic}

\usepackage{graphicx}

\usepackage{textcomp}

\usepackage{xcolor}

\usepackage{url}

\usepackage[T1]{fontenc}
\usepackage[utf8]{inputenc}

\usepackage{booktabs}
\usepackage{tabularx}

\newtheorem{proposition}{Proposition}

\begin{document}

\title{Defining Decentralization: An Ontological Perspective}

\author{
    \IEEEauthorblockN{Jakub K. Szel\k{a}g}
    \IEEEauthorblockA{\textit{School of Computing}\\
    Newcastle University\\
    Newcastle upon Tyne, UK\\
    J.K.Szelag2@newcastle.ac.uk}
    \and
    \IEEEauthorblockN{Aydin Abadi}
    \IEEEauthorblockA{\textit{School of Computing}\\
    Newcastle University\\
    Newcastle upon Tyne, UK\\
    Aydin.Abadi@newcastle.ac.uk}
    \and
    \IEEEauthorblockN{Mohammad Naseri}
    \IEEEauthorblockA{\textit{Flower Labs}\\
    mohammad@flower.ai}
}

\maketitle

% An idea to strengthen the paper: So far we have issues with portraying ownership, which might be attacked by the reviewers. potentially include a "realization -> center" divide, could be interesting and strengthen our claims. Also rewrite related works, although I've written most of it myself- ai tools flag it as "Ai written". 

% initial idea: \mu(u) = |\delta_u(R_u)| where \delta_u maps from raw realizations to the number of C_u centers? Look into that- would really strengthen the paper! 

\begin{abstract}
% Rewrtitten the abstract so it incorporates the our previous discussion
Decentralization as a concept in computer science has existed for over half a century. Despite its fundamental role across domains such as security, distributed computing, artificial intelligence, cloud infrastructures, and Internet of Things (IoT) architectures, there remains no universally accepted definition of decentralization applicable across computer communication systems. This absence has become increasingly problematic with the emergence of decentralized AI and machine learning paradigms, including collaborative training, distributed inference, blockchain-based, and agentic AI, where decentralization is often treated as a core design objective. Meanwhile, existing approaches frequently conflate decentralization with related notions such as distribution of trust or specific implementation paradigms. This ambiguity creates inconsistencies in system analysis, limits comparability between works, and weakens the rigor of formal reasoning surrounding communication architectures and protocol design. In this work, we define this research gap as the \textit{Decentralization Problem}.

We analyze the formal-semantic, epistemological, and pragmatic foundations of decentralization, showing why existing definitions remain context-dependent, technology-specific, or conflated with distribution. From this analysis, we derive the requirements for a transferable formal treatment and introduce a graph-based ontology that treats raw realization multiplicity, topological distribution, and decentralization with respect to a declared anchor (such as ownership or authority) as distinct properties. Anchors are interpreted through graph-realistic center regions, allowing a subject to be distributed across several vertices while remaining centralized with respect to a single center. The framework supports multidimensional, anchor-relative evaluation through two analytical metrics: Void Tolerance and Imperviousness. We implement the framework as a browser-based application supporting interactive modeling, automated classification, metric computation, and large-scale deterministic simulation. Applications to federated learning and blockchain models demonstrate a cross-domain capabilities for producing comparable assessments when the graph abstraction, evaluation profile, and center interpretations are disclosed.

\end{abstract}


\section{Introduction}

% To strenghten the interest of reviewers: Mention a good example comparing the original work commonly used in the literature (baran) and the definitions behind decentralized ML paradigms, and how decentralized AI is viewed (Give vanilla FL as the main example) - done

% Simplify the motivation, why is it interesting,why is it a problem etc. - abstract

Decentralization has become one of the defining architectural principles in modern computer communication systems. The concept is foundational across security, distributed computing, and increasingly artificial intelligence. Despite that, there exists a striking inconsistency between how decentralization is represented and how it is described in many modern systems.

A useful illustration emerges by contrasting one of the earliest and most influential models of decentralization with a widely adopted machine learning paradigm. Baran’s \cite{baran_distributed_1964} seminal work on distributed communications distinguishes centralized, decentralized, and distributed networks as structurally different graph topologies, with decentralization explicitly characterized through intermediate hierarchical connectivity and distribution represented through fully peer-connected redundancy \cite{baran_distributed_1964}. This structural framing has shaped decades of systems literature and remains a common conceptual reference point.

By contrast, modern decentralized machine learning literature frequently applies the term \emph{decentralized} to architectures whose structural properties differ substantially from Baran’s taxonomy. Federated Learning (FL) serves as a particularly useful example, one we will consistently refer to. In vanilla FL \cite{mcmahan_communication}, a central server coordinates model aggregation while client devices train locally and periodically transmit updates to the server. This approach is widely described as privacy-preserving and decentralized because data remains distributed across participants, allowing for clients to keep their own datasets private. From a communication-graph perspective, however, the system remains structurally centralized: clients depend on a single aggregation point, communication paths are mediated through the server, and the server represents a coordination bottleneck and a trust authority. This is often represented as a star-topology, which through the lens of \cite{baran_distributed_1964} is considered to be centralized. Interestingly, despite the fact that the FL as a whole is seen as one of the more promising decentralized ML paradigms, further classifications follow in the FL literature breaking it down further into ``Centralized FL'' and ``Decentralized FL'' \cite{yuan_decentrazlied_2024}.

This discrepancy highlights a crucial question, central to this work. If a canonical systems model classifies a topology as centralized, while field-specific literature labels an analogous architecture as decentralized, then questions arise as to the exact meaning of decentralization. It is clear that Baran's structural point of view significantly differs from much of, for example, ML literature where decentralization is inferred from placement of data or model training. As much as these are important properties for ML, they are only specific dimensions that do not fully encapsulate decentralized communication structure even in the field of ML, provided that we may have to consider for other dimensions (e.g., security-enhancing solutions). As a result, identical architectures may be described as decentralized and centralized as a whole \textit{in the same context}. This ambiguity raises a question regarding the formal semantics of these terms, and how can we universally address them. 

These questions have become increasingly important with the emergence of decentralized AI systems, including federated learning, swarm intelligence, collaborative model training, distributed inference pipelines, and blockchain-assisted AI \cite{DecFL_Beltan_2023, Aishwarya_swarmIntelligence_2023, decInference_wang_2026, Blockchain_AI_Salah_2019, collabML_Saif_2025}. Across these domains, decentralization is frequently presented as a core design objective and as a source of resilience, privacy, autonomy, and trust minimization. However, the term itself is often used interchangeably with distribution, locality, or absence of centralized storage, even when underlying coordination mechanisms remain structurally centralized.

This creates an unusual tension, as decentralization is routinely invoked as both a design objective and an analytical property, and systems are often described, compared, or evaluated in terms of how decentralized they are. All while lacking a formal definition that specifies what decentralization is independent of a particular technology or application domain, making these comparisons often incompatible across literature and conclusions from evaluations inconsistent. In practice, decentralization is frequently treated as self-evident, that is, inferred through intuition, approximated through domain-specific metrics, or assumed through the absence of a surface-level, and easily identifiable central authority. While useful in narrow contexts, these interpretations do not generalize and often reinforce previously listed concerns even in the same context.

This lack of formalization gives rise to several unresolved research problems.

\textbf{Problem 1: Definitional Ambiguity.}  
There is no universally accepted, transferable definition that both supplies invariant truth conditions and exposes the context under which they are evaluated. Existing work often relies on examples, system-specific heuristics, or associated properties such as resilience and trust minimization. As a result, one system may be called decentralized under one implicit interpretation and centralized under another.

\textbf{Problem 2: Conceptual Confusion with Distribution.}  
A persistent ambiguity throughout the literature is the interchangeable use of \textit{decentralization} and \textit{distribution}. Although the terms are often treated as equivalent, they refer to fundamentally different properties: one concerning relational structure and dependency, the other concerning allocation or placement of components. Without a formal distinction, analytical conclusions derived from one are frequently attributed to the other, which obscures both.

\textbf{Problem 3: Lack of Structural Formalism.} 
Many existing approaches evaluate decentralization through node-centric quantities such as participant counts, ownership concentration, or resource allocation. While informative, these approaches abstract away the structure of relationships between entities. This presents a fundamental limitation, as decentralization is inherently relational, it emerges through dependencies, communication paths, and structural constraints connecting participants rather than through participant attributes in isolation.

\textbf{Problem 4: Absence of Transferable Quantification.}  
Even where decentralization is quantified, metrics are often tied to particular architectures, most notably blockchains. This limits reuse across peer-to-peer networks, federated learning, and decentralized AI coordination systems, and leaves the compatibility conditions for comparison unstated.

Addressing these challenges requires a formal treatment where rules are transferable across technologies and application domains while governance, ownership, authority, and similar interpretations remain explicit inputs. The remainder of this paper develops such a framework within a graph-realistic scope.

\section{Our Contributions}

With that said, this work addresses the \textit{Decentralization Problem}: the absence of a formal and transferable rule for evaluating decentralization across computer communication systems without concealing the context that gives a decentralization claim its meaning. We first examine why the concept has resisted formalization, then derive a graph-theoretic ontology and apply it across heterogeneous system models. Our contributions are:

\begin{enumerate}
    \item\textbf{We expose why existing approaches do not provide a transferable definition.} We show that prevailing approaches are context-dependent, technology-specific, or conflate decentralization with distribution, governance, or resource allocation, identifying the obstacles to a reusable formal treatment.
    
    \item\textbf{We identify the requirements and scope of a cross-domain definition.} These include a transferable rule, explicit relational structure, compatibility with formal reasoning, and preservation of contextual interpretation through declared subjects and anchors.

    \item \textbf{We introduce a graph-based ontological definition of decentralization for computer communication systems.} We propose an ontology-level formalization designed as a general conceptual foundation independent of any application domain. It distinguishes a subject from the anchor under which centralization is interpreted (for example, ownership or authority), represents the corresponding centers extensionally as vertex regions, supports multiple subject-anchor evaluations, and is directly applicable to arbitrary \textit{graph-representable} systems.

    \item \textbf{We formally separate decentralization from distribution and redefine centralization through this framework.}  
    A central result of this work is a rigorous semantic distinction between decentralization and distribution within graph-realistic setting, resolving the ambiguity present throughout the literature. The functions $\delta$, $\lambda$, and $\mu$ respectively count realization particulars, supporting vertices, and center regions reached under an anchor. Within the same framework, we show that centralization emerges naturally as the degenerate zero-dimensional case of decentralization, establishing a cleaner and more principled relationship between the two.

    \item \textbf{We extend the definition with analytical metrics beyond classification.} Based on our ontology, we introduce two graph-based subject-specific metrics: \textit{Void Tolerance} and \textit{Imperviousness}, which measure the impact of nodes and connections respectively. Together they produce a decentralization vector that enables direct comparison of systems and exposes drawbacks that scalar-only measures cannot capture. On top of that, we base our reasoning in formal semantics to argue for why these metrics (specifically, their aggregate), can be used to describe a system as more or less decentralized from one another.

    \item \textbf{We validate the framework on heterogeneous real-world systems and demonstrate transferability.}  
    We instantiate the ontology on fundamentally different communication architectures, including blockchain and federated learning, and show that the framework yields consistent classifications and meaningful analytical insight where existing definitions become incomplete or contradictory. This demonstrates that the proposed ontology can function as a universal formal foundation for decentralization across computer systems.

    \item \textbf{We provide a sandbox implementation for decentralization analysis.} Based on our analytical evaluation metrics and ontological construction. We implement a prototype tool that can be used to evaluate arbitrary systems encoded in graph representations, providing an ontological inference regarding the status of decentralization, as well as exact values for Imperviousness and Void Tolerance.
\end{enumerate}

%The goal of this work is to establish a rigorous and transferable definition of decentralization for arbitrary computer communication systems, independent of implementation paradigm, governance assumptions, or application context.

\section{Related Works}\label{sec:related_works}

\subsection{Structural Descriptions}

Across network theory, distributed systems, and blockchain research, decentralization has become a remarkably productive object of study \cite{lui_2026_blockchainbased, yuan_decentrazlied_2024, Bodo2021Decentralisation}. The resulting literature has developed our informal understanding of decentralized architectures and furnished an increasingly sophisticated analytical toolkit. What was not successful, however, was a comparable agreement on the definition of the concept as decentralization is still rarely defined formally. Given our earlier discussion, it is more often inferred from intuition, reconstructed from context, or identified with the properties of a particular application. The literature is therefore rich in uses of decentralization, but fragmented in its formalizations of it.

Much of the vocabulary begins Baran's influential three-part typology (centralized, decentralized, and distributed) remains one of the field's earliest, and most recognizable structural descriptions \cite{baran_distributed_1964}. The work introduced redundancy level as the minimum number of links needed to preserve connectivity and used probabilistic models to evaluate resilience under node and link failures. That separation between nodes and links is consequential, as we will see later,  already indicating that structure cannot be reduced to a count of participants.

Nevertheless, it is not unreasonable to state that a topology taxonomy is not yet a definition. Baran's categories are conveyed through representative diagrams, examples, and intuitive argument, rather than through necessary and sufficient conditions for decentralization. Decentralization consequently appears as a recognizable pattern, not a rigorously specified property. The decentralized case is also rendered hierarchically. In effect, the illustration couples decentralization to hierarchy without formally defending that coupling, and without establishing that it extends to decentralized systems beyond the depicted architecture \cite{baran_distributed_1964}.

\subsection{Consensus and Trust Frameworks}

From topology, the literature turns to agreement under adversity. The ``Byzantine Generals Problem'' and the Byzantine fault-tolerance literature that followed brought exact reasoning to coordination among potentially adversarial participants \cite{lamport_byzantine_1982}. In the standard setting, these results establish that reliable consensus cannot be guaranteed once adversarial participants reach or exceed $\frac{1}{3}$ of the total. They thereby connect decentralized operation to provable guarantees about trust and fault tolerance.

Here, however, rigor is directed at a neighboring question. The formal object is the feasibility of consensus under stated adversarial assumptions, not decentralization itself. Consensus is specified, whereas decentralization remains the operating context in which consensus must be achieved. Its structural basis, and indeed its epistemological meaning, is left implicit \cite{lamport_byzantine_1982}.

Later introduced CAP theorem sharpens a different boundary. Consistency, Availability, and Partition tolerance cannot all be guaranteed simultaneously in a distributed system according to \cite{brewer_towards_2000}. By making this trade-off explicit, CAP replaced a simpler picture of distribution with one structured by unavoidable design choices. Yet it still begins after the relevant conceptual assumption has been made, that being a preexisting view of a distributed architecture. Despite its significance, it does not tell us what first makes that system decentralized. Its influence on the analysis of system properties is therefore substantial, while its contribution to defining decentralization is indirect at most.

Blockchain changed both the technological landscape and, coincidentally, the vocabulary of this debate \cite{nakamoto_bitcoin_2008}. One can observe that in Bitcoin's original work, and even more strongly in the literature that followed, decentralization became associated with removing the trusted third party (even though the original paper does not use the term ``decentralized''). That association proved powerful as it made blockchain the predominant setting in which decentralization was discussed, operationalized, and evaluated.

The cost of that influence was conceptual narrowing in a form of removing an intermediary allowing to be an expression of decentralization, despite being technology and context-specific. Subsequent literature often allowed that distinction between distribution and decentralization to collapse, using decentralization as if it were interchangeable with peer-to-peer distribution or trust minimization. What remains missing was a formal account of where one concept ends and the others begin.

\subsection{Quantitative Proxy Measures}

Once decentralization became an explicit design objective (or a desirable property of a system), measurement followed. Leading with \textit{Nakamoto coefficient} as perhaps the clearest example, recording the minimum number of entities whose compromise would be sufficient to obtain systemic control over a blockchain network \cite{srinivasan_quantifying_2017}. Still, the metric begins by taking the blockchain architecture for granted. It represents decentralization through control over particular subsystems and derives that control from node counts alone. What it measures, therefore, is concentration within a specified architecture and not decentralization as a structural property transferable across architectures. Approaches based on the effective distribution of power among protocol participants reach a similar boundary, given that their quantities depend on blockchain-specific consensus rules, participant roles, and resource assumptions \cite{kwon_impossibility_2019}.

Other studies broaden the quantitative toolkit by borrowing from adjacent disciplines. The \textit{Gini coefficient} captures inequality in resource allocation, and \textit{Shannon entropy} captures uncertainty or dispersion across participants \cite{lin_measuring_2021}. Yet, neither resolve the underlying definitional problem. They reveal how resources or participation are distributed among nodes, but not what decentralization itself is.

More importantly, these measures see attributes more readily than relations. Holdings, participation rates, and node shares are visible; the dependencies connecting the corresponding entities are abstracted away. For decentralization, that abstraction is decisive, as it is known in the literature that this property is expressed through relationships and dependencies among participants, not through participant attributes considered one at a time---making decentralization relational \cite{rossi_towards_2019}.

\subsection{Systematization}

A further step becomes visible when the literature moves beyond measuring a single instantiation and attempts to systematize decentralization across a class of systems or contexts \cite{Troncoso_2017, shahsavari_2022_quantifying}. These works bring the definitional problem closer to the foreground, although they approach it from markedly different directions.

The authors of \cite{Troncoso_2017} draw an informal epistemological distinction between distributed and decentralized systems. In their account, decentralization is relational, while distribution operates as its direct prerequisite. This is an important shift allowing the analysis to no longer treat the two terms as uncomplicated synonyms.

However, the framework is scoped from the outset to privacy and places particular weight on trust among authorities. Such a focus may be sufficient within parts of the security literature, but it creates tension elsewhere. Systems commonly regarded as decentralized may contain elements that are not themselves distributed, take for example cryptocurrency nodes, which may be physically clustered on the same host, in one data center, or within a single cloud region \cite{gencer_2018_decentralization_bitcoin}. The privacy-centered scope also travels less easily to discussions of decentralization outside security. Ultimately, decentralization is considered \textit{en bloc} as the principal question behind how it affects privacy, not what decentralization is in the first place \cite{Troncoso_2017}.

Approaches to the problem from the other end where quantification takes priority \cite{shahsavari_2022_quantifying} provide us with an analytical framework that is confined to blockchain networks, often instantiated as unstructured peer-to-peer architectures. Within that setting, however, distribution and decentralization are not explicitly distinguished and therefore never become objects of contrast.

Decentralization instead appears through the forms it assumes in blockchain systems, rather than as a formal concept capable of standing independently of them. Once again, the literature succeeds in illuminating the properties, manifestations, and measurable consequences of decentralization while leaving its definition unresolved. That recurring displacement (from the concept to its effects) provides a further motivation for the research question pursued here \cite{shahsavari_2022_quantifying}.

\subsection{Summary}

The trajectory is both substantial and incomplete. From topology and resilience, through consensus and trust, to architectural trade-offs and quantitative concentration, the literature has formalized many of the phenomena that surround decentralization. What repeatedly becomes rigorous is an associated property, such as connectivity, fault tolerance, consistency, intermediary removal, resource dispersion, or control. Decentralization itself remains largely descriptive, context-bound, or inferred from those proxies \cite{baran_distributed_1964, lamport_byzantine_1982, brewer_towards_2000, nakamoto_bitcoin_2008, srinivasan_quantifying_2017, Troncoso_2017, shahsavari_2022_quantifying}. It is this persistent gap between formalizing what properties correlate with decentralization and formalizing decentralization itself that motivates the ontology introduced in this work.

\section{Preliminaries}

\subsection{Adjectival Dimensionality of Decentralization}

The root of the problem becomes much clearer once we view decentralization through the lens of formal semantics, as decentralization is a multi-dimensional gradable adjective, where what it attributes cannot be fully confined in a single dimension. A rudimentary work in this direction by Sassoon \cite{sassoon_typology_2013} shows that multidimensional adjectives require a dimension parameter, that is, when a speaker uses an adjective such as "healthy", the relevant dimensions are often left to context, and some binding operation determines how those dimensions jointly produce a single interpretation. This is directly relevant to decentralization, given that in much of the literature, a system is described as decentralized only with respect to some implicitly selected set of dimensions. However, the operation by which these dimensions are bound together is often left unspecified. As a result, different works may appear to disagree about decentralization while in fact relying on different contextually selected dimensions, or on different ways of combining those dimensions.

Following Sassoon's work, multidimensional adjectives may be interpreted conjunctively, disjunctively, or in a mixed manner: (a) an entity may need to satisfy all relevant dimensions, (b) one relevant dimensions, or (c) a pragmatically determined subset of them in order to fall under the adjective \cite{sassoon_typology_2013}. Decentralized appears to behave as a mixed adjective in this sense. Consider a system whose communication is decentralized, but in which aggregation authority or data remain centralized. We then begin to ask ourselves a question if that system should or shouldn't be described as decentralized and why. Conversely, consider a system with decentralized data, but whose infrastructure depends on a small number of coordinating entities. Whether that system is decentralized or not depends on both the degrees reached along each dimension, but also on which dimensions are taken to be relevant and how they are bound into a single judgment.

For the formal treatment in the following sections, each relevant dimension is made explicit as a subject-anchor pair, that is, what is projected over the graph, and with respect to which interpretation its centers are counted. This permits the same subject to be evaluated under ownership and authority without equating either anchor with the subject itself.

\subsection{Distribution vs. Decentralization}

The academic community generally agrees that decentralization and distribution are distinct concepts. Distribution commonly refers to the placement or dispersion of system elements across multiple locations, whereas decentralization concerns the structure, organization, and governance of relationships between entities within a system \cite{rossi_towards_2019}. This distinction is not secured solely by using two names for counts. For example, physical distribution may remain within one effective center (e.g., all entities fall under one ownership).

If we combine this with previously discussed formal semantics, it helps explain why decentralization is difficult to compare across contexts. Both "decentralized" and "not decentralized" may be asserted with respect to some dimension of the same system, just as "centralized" and "not centralized" may also hold relative to different dimensions. Distribution, by contrast, is more commonly used disjunctively in technical discourse, where a system may be called distributed if at least one relevant class of components or resources is placed across multiple locations or entities \cite{rossi_towards_2019}. Decentralization does not usually permit the same looseness without further clarification. Its dimensions must be specified, and the standard for satisfying them must be made explicit. This is precisely where an ontological framework becomes useful, it provides a cross-contextual structure for identifying the relevant dimensions of decentralization and for making explicit the operation by which those dimensions are combined into a coherent assessment.

\subsection{Formal vs. Informal Reasoning in Decentralization}

Broadly speaking, reasoning is defined as the a cognitive process of drawing inferences, evaluating arguments, and forming conclusions from premises. The main distinction between formal and informal reasoning is thus drawn by what governs them both. \textit{Formal reasoning} requires explicit rule-bound systems (in our context provided by mathematical logic) where validity of arguments is entirely independent from the context, prioritizing form over contents. Requiring closed systems that are precisely defined any by their very nature incomplete \cite{godel_uber_1931}. \textit{Informal reasoning} however, operates within the domain of natural language where formal logical validity alone is deemed as insufficient (or inapplicable \cite{Toulmin1958-TOUTUO-8}).  

Within this framing, the term decentralization exposes the limitations of a strict formal/informal divide. Formally, decentralization can be specified in terms of graph-theoretic properties of distributed networks and/or the absence of a singular coordinating authority, thereby lending itself to representation within rule-governed systems. However, its practical invocation in discourse about system design frequently extends beyond these formal properties to include informal claims concerning robustness, transparency, governance structure, trust, and more. Showcasing how formal technical concepts are often interpreted through informal lens which incorporates contextual judgment. On that basis decentralization serves as a great example of how reasoning in computer science oscillates between formally tractable descriptions and informally grounded assessments, thereby reinforcing the view that the boundary between formal and informal reasoning is functionally dependent in applied contexts.

Although this distinction might appear as self-evident from the contemporary standpoint in the field of computer science, it is crucial to highlight for our argument. Our arguments that link decentralization, ontologies, and formal logic, will be inspired by the claims of Finocchiaro \cite{Finocchiaro_2005, Finocchiaro2005-FINAAA-4}; specifically, those which establish that this divide is complementary rather than being mutually exclusive.  

\section{The Decentralization Problem}

\subsection{Informal Reasoning as a Basis}

In practice, the informal component of reasoning often becomes the \textit{de facto} basis upon which ``decentralization'' is defined and assessed \cite{king_centralized_1983, muller_data_2017, di_bona_concept_2023, yanagihara_cross-referencing_2021, sun_duopoly_2024, luo_token_2026} where discourse frequently defaults to intuitively accessible but weak notions such as ``absence of centralization'' \cite{tsigkanos_2019_towards}, or resorts to their own interpretations of the term \cite{yanagihara_cross-referencing_2021}. The term acquires a degree of flexibility that allows it to function as a normative label as much as a technical descriptor. This presents a structural problem as evaluative judgments about decentralized systems may be formed on the basis of rhetorically important but analytically weak criteria. In this sense, ``decentralization'' becomes a case study in how informal reasoning can silently dominate conceptual formation.

\subsection{Pragmatic Multi-Dimensionality of Decentralization}

Previous works that focus on questions around defining decentralization in different contexts have converged on a shared intuition that decentralization is not a binary property, but multi-dimensional. Notably, \cite{rossi_towards_2019} provides an extensive treatment of this problem, drawing lessons directly from blockchain networks to present a framework that breaks down decentralization across several distinct dimensions, each of which can be independently more or less centralized. While the scope of that work remains largely confined to blockchain infrastructures and the framework itself is conceptual rather than formal, it nonetheless establishes a precedent that rigorous attempts to describe decentralization must account for multiple axes of evaluation.

Across broader literature, there is a growing acknowledgment that decentralization resists a single, unified measure. As noted in \cite{alghuried_blockchain_2025}, there exists no exact mathematical definition of decentralization; rather, quantitative approaches such as the Gini coefficient are employed to approximate the degree to which a system tends toward one end of a spectrum or the other. As such, decentralization is treated not as an absolute state, but as a scalar property. 

The lack of formalization leaves us without a principled basis for cross-system reasoning. A transferable account must therefore satisfy three requirements. First, it must be multidimensional in the subjects under evaluation. Second, it must state the anchor (for example ownership or authority) under which centers are distinguished, rather than treating that interpretation as another subject or leaving it as implicit. Third, it must separate a logical truth condition for each declared pair from any subsequent gradable evaluation of topology-dependent extent. The framework below uses $\mu$ for the binary anchor-relative boundary and retains Void Tolerance and Imperviousness as explicitly non-unique analytical degree functions.

\subsection{Abstraction Level for Approaches}

% We address the issues above in culmination with the following approach:
% ^ the level of abstraction for the previous works that try to define and formalize decentralization was too "low", that is, it was too constrained to the given contexts provided the informal reasoning as the basis. Operating more on levels of taxonomies rather than formal definitions despite their attempts, this is shown by the fact that more recent attempts try to formalize decentralization *for the purpose of* a given context authors want to review.

The works surveyed in Section \ref{sec:related_works}, taken together, reveal a recurring and systematic limitation. Namely, the level of abstraction at which decentralization has been approached has remained insufficient, constrained by the specific contexts from which each contribution emerged and by the informal reasoning that has served as its conceptual basis. Rather than producing general formal definitions, the existing literature has largely produced taxonomies and context-specific measures, a tendency we can observe from most recent attempts to formalise decentralization doing so specifically for the purpose of evaluating a particular system or setting the authors wish to examine \cite{srinivasan_quantifying_2017, kwon_impossibility_2019, lin_measuring_2021}. The result is a body of work in which the definition follows from the application, rather than the application following from the definition. It is precisely this inversion that the present work seeks to correct by approaching decentralization at a level of abstraction sufficient to accommodate a formal treatment.

\subsection{Culminating the Decentralization Problem}

% We conclude the problem in the following manner:
% ^ There is a clear gap where decentralization, despite being used in formal settings, is based in informal reasoning when discussed in different contexts, making it a non-transferrable concept among different applicaitons. 
% ^ There are methodological issues with the recent approaches, despite the improvements they make in acknowledging multi-dimensionality, because in turn, they often disregard the relational aspect and connections between entities within the communication systems. Marking a clear distinction between approaches before and after the introduction of bitcoin, since it became a conceptual "golden-standard" for decentralization within the literature. As provided by our argumentation, this is incorrect given the omissions it makes.
% ^ These methodological issues within the level of abstraction, as well as the infromal basis for the reasoning is the reason why we have the "Decentralization Problem"

Despite its prevalent use in formal settings across multiple domains of computer science, decentralization remains a concept grounded in informal reasoning, defined contextually and inconsistently such that it does not transfer reliably between applications or system classes. The methodological issues identified in the more recent literature compound this further: while later approaches acknowledge the multi-dimensional nature of decentralization as an improvement over earlier binary characterisations, they do so while simultaneously disregarding the relational aspect of communication systems, reducing their analyses to the properties of individual entities and omitting the connections between them. This tendency is closely tied to the significant influence of blockchain-based systems on the literature, which following their introduction effectively became the implicit conceptual standard against which decentralization is defined and measured. As we have already demonstrated, this is an insufficient basis for a general treatment of the concept, given the significant structural omissions it entails. 

It is the combination of these factors, namely the absence of more formal and transferable definition, the methodological inconsistencies in quantification, and the over-reliance on a single application context as a conceptual baseline, that altogether constitute what we refer to in this work as the \textit{Decentralization Problem}.

Table~\ref{tab:related_work_comparison} consolidates our analysis by comparing representative approaches against the requirements derived from the above. Specifically, the comparison considers whether each approach: (P1) defines decentralization itself rather than an associated property; (P2) explicitly distinguishes decentralization from distribution; (P3) represents decentralization through relationships and dependencies between entities; and (P4) provides quantification transferable beyond its original application domain (these are based on earlier established research problems). We additionally report whether an approach explicitly supports multiple dimensions of decentralization and whether it provides an ontological foundation.

\begin{table*}[t]
\centering
\caption{Comparison of representative approaches against the requirements derived from the Decentralization Problem.}
\label{tab:related_work_comparison}
\scriptsize
\setlength{\tabcolsep}{2.5pt}
\renewcommand{\arraystretch}{1.22}
\begin{tabular}{p{0.14\textwidth} p{0.25\textwidth} c c c c c c}
\hline
\textbf{Work} & \textbf{Primary focus and scope} & \shortstack{\textbf{Formal}\\\textbf{definition}\\\textbf{(P1)}} & \shortstack{\textbf{Decentralization/}\\\textbf{distribution}\\\textbf{distinction (P2)}} & \shortstack{\textbf{Relational}\\\textbf{structure}\\\textbf{(P3)}} & \shortstack{\textbf{Transferable}\\\textbf{quantification}\\\textbf{(P4)}} & \shortstack{\textbf{Multi-}\\\textbf{dimensional}} \\
\hline
Baran \cite{baran_distributed_1964} & Structural network typology and resilience under node and link failures. & -- & $\triangle$ & \checkmark & -- & -- \\
Lamport et al. \cite{lamport_byzantine_1982} & Consensus feasibility and Byzantine fault tolerance under adversarial participation. & -- & -- & $\triangle$ & -- & --  \\
Brewer \cite{brewer_towards_2000} & Trade-offs among consistency, availability, and partition tolerance in distributed systems. & -- & -- & $\triangle$ & -- & -- \\
Nakamoto \cite{nakamoto_bitcoin_2008} & Peer-to-peer electronic cash and removal of a trusted third party. & -- & -- & $\triangle$ & -- & --  \\
Srinivasan and Lee; Kwon et al. \cite{srinivasan_quantifying_2017,kwon_impossibility_2019} & Blockchain-specific concentration, subsystem control, and effective protocol power. & -- & -- & -- & -- & $\triangle$ \\
Lin et al. \cite{lin_measuring_2021} & Blockchain decentralization proxies based on resource inequality and participation dispersion. & -- & -- & -- & -- & $\triangle$ \\
Rossi et al. \cite{rossi_towards_2019} & Conceptual and multidimensional systematization of decentralization, primarily in blockchain infrastructures. & -- & \checkmark & \checkmark & -- & \checkmark \\
Troncoso et al. \cite{Troncoso_2017} & Privacy- and trust-oriented distinction between distributed and decentralized architectures. & -- & \checkmark & \checkmark & -- & -- \\
Shahsavari et al. \cite{shahsavari_2022_quantifying} & Analytical quantification of decentralization in blockchain peer-to-peer networks. & -- & -- & -- & -- & $\triangle$ \\
\textbf{Our work} & \textbf{Domain-independent, graph-based ontology with subject-specific classification and analytical metrics.} & \checkmark & \checkmark & \checkmark & \checkmark & \checkmark \\
\hline
\end{tabular}

\vspace{1mm}
\begin{minipage}{0.98\textwidth}
\footnotesize
\textit{Legend:} \checkmark\ indicates that the requirement is explicitly addressed; $\triangle$ indicates that it is addressed partially, descriptively, or only within a restricted context; ``--'' indicates that the requirement is not addressed as part of the work's treatment of decentralization. P1--P4 correspond respectively to the research problems: definitional ambiguity, confusion with distribution, absence of structural formalism, and absence of transferable quantification. The table evaluates each work only with respect to its treatment of decentralization; it does not assess the broader significance of its original technical contribution.
\end{minipage}
\end{table*}

% Repeating myself again
% The comparison illustrates that prior work addresses important properties associated with decentralized systems, including topology, resilience, consensus, trust, concentration, and multidimensionality. However, these contributions generally satisfy only a subset of the requirements identified by the Decentralization Problem. Structural approaches retain descriptive definitions, formal distributed-systems results analyze properties other than decentralization itself, and quantitative approaches remain predominantly node-centric or blockchain-specific. The proposed framework differs by integrating the four requirements within a single ontological treatment and defining decentralization independently of a particular application, separating it formally from distribution, modeling it as a relational and subject-specific property, and supporting cross-domain quantification through analytical metrics.

\section{Ontology as the Foundation}

% Ontologies are formalised conceptualizations of knowledge, allowing for formal descriptions of "non-logical" concepts such that True/False statementes can be evaluated and quantified. 
% They sit within the formal domain following the development of formal reasoning: propositional logic -> Predicate logic -> description logics/languages -> finally onto ontologies.
% BUT they are also capable of encapsulating concepts that can appear in the informal reasoning (specifically, the non-logical concepts), making it an appealing tool to address this problem

Defining decentralization and addressing the Decentralization Problem requires a framework capable of operating simultaneously within the domain of formal reasoning and across the breadth of contextual variation that the concept exhibits in practice. Ontologies, as formalized conceptualizations of knowledge, present themselves as a natural candidate for this purpose. Meanwhile as a tool, they provide a rigorous approach within which concepts and their relationships can be defined such that logical statements about them can be decidable and quantifiable. Crucially, however, ontologies are not restricted to purely logical primitives; they are equally capable of encapsulating non-logical concepts, that is, concepts whose meaning is not fully reducible to formal symbols alone but which can nonetheless be given precise structural descriptions within an ontological framework. More importantly, ontologies allow for context-dependence while still being regarded as formal. These properties operate with the rigor demanded by formal reasoning while remaining expressive enough to capture concepts that originate from informal discourse, making ontologies a particularly well-suited tool for addressing the  methodological gaps identified in our earlier discussion.

% We talk here about the current use of ontologies in the context of decentralization (refer to my own citations and notes)
% Interestingly, the use of formal ontologies for decentralization suffers from similar problem as previously discussed decentralization in the terms of abstraction level. Concretely, as appraoches to defining decentralization are too context-specific in order to define the concept of decentralization in the terms of *generalization* of decentralization definitions, ontologies used to define decentralized settings are also too context-specific in order to achieve that, as they often try to formalize conceptual knowledge within a particular domain that is more in line with a specific application rather than decentralization in general 

The application of ontologies to decentralized settings is not without precedent. Existing work has employed ontological frameworks across a range of domains in which decentralization plays a central role, including decentralized energy systems \cite{ont_wu_ontology_2021, ont_kaya_decent_2021}, eco-industrial information management \cite{ont_zhou_ontology_2018}, decentralized production control in cloud manufacturing \cite{ont_katti_ontology-based_2020}, and self-organising multi-agent systems \cite{ont_liu_adaptive_2021}. Broader cross-disciplinary treatments have also sought to map the conceptual landscape of decentralization across domains \cite{ont_hoffman_toward_2020}, and the relationship between decentralization and the organisation of knowledge itself has been explored at a philosophical level \cite{ont_halpin_decentralization_2016}. 

Existing ontologies in this space are generally developed for particular domains such as energy trading, manufacturing, or collective decision-making rather than as transferable treatments of decentralization. The resulting context dependence is not itself a defect given that every concrete claim needs interpretation. What we identify here as missing is a common rule that preserves and exposes that interpretation instead of burying it in a domain-specific label. Our use of subjects, anchors, and graph-realistic center regions targets that level of abstraction while remaining explicitly conditional on the declared context.

\section{Our Definition}

% Graph-based ontology for decentralisaiton, talk about usage of graphs
% Describe the aim as an introductory paragraph, mention ontology101 as the main reference (but again there is more so I'll take care of that later, the methodologies across the literature are idosyncratic sooo, mention gruber's and guarino's work too tbh)

To make our approach easier to follow, we organize our ontology-based definition into four layers: (a) \textbf{Ontological Layer}, where the core ontology will be defined with conceptual commitments, competency questions, signature, and taxonomy, (b) \textbf{Logical Layer}, where formal semantics, axiomitization, and reasoning properties will be specified, (c) \textbf{Modeling Layer} where we bind the graph-based topological representation, and (d) \textbf{Analytical Layer} where we introduce graph-based metrics as measures for decentralization and evaluate them on previous instantiations to showcase the following computational capabilities when using our definition. Aligning our approach as much as possible with previous works focusing on  ontology development such as \cite{guarino1998, ontology101}, while using standard notation for ontology and description logic \cite{baader_description_2007, arp_building_ontologies}. Similarly to those works, we also note that ontology development methodologies are known to be idiosyncratic (i.e. there is no single standard approach).

In our case, these layers are necessary to function as a definition. Provided that in formal reasoning direct definitions are bounded to the underlying context in order to make precise logical inferences. For us to have the capacity to produce truth-value expressions we are required to construct an ontological foundation, around which we will be able to formalize axioms governing inferences. Simultaneously, if we want to elevate this foundation onto practical systems and introduce quantifiable metrics based on the ontology, we need to specify how do we model our ontology for later analytical evaluation. Figure \ref{fig:roadmap} visualizes how these layers are connected altogether.

\begin{figure}[ht]
    \centering
    \includegraphics[width=0.5\textwidth]{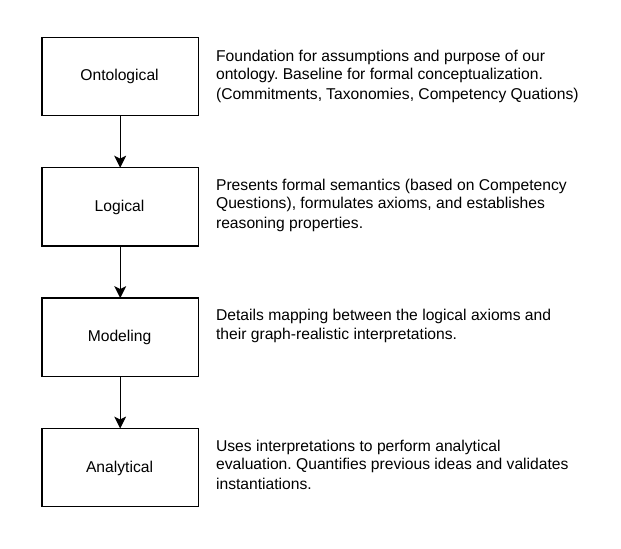}
    \caption{Ontology Roadmap.}
    \label{fig:roadmap}
\end{figure}

% ??We lean towards structural relationalism with contextual operationalism with graph-theoretic pragmatism

% Talk about complexity reduction at some point (maybe not here) to ensure its viably usable (also examples in instantiations will be of great use here)

\section{Ontological Layer}

This section aims to first formalize the core ontology behind decentralized and distributed systems (without using ``decentralization'' and ``distribution'' as primitive concepts), from which we can derive a formal definition for decentralization.

\subsection{Core Ontology}

%\subsection{Commitments, Signature, and Competency Questions}

% This section is mainly reconstructed/shortened from 2-3 pages of writing from main.tex, we only remove the aspect of proofs and context (as ontological concepts)
\noindent\textbf{Ontological Commitments.} Conceptually, the first and most fundamental commitment is to a \textit{graph-theoretic} realism about systems. A system is taken to exist as a configuration of computational entities and their interconnections, formally representable as a non-empty graph $G = (V, E)$. This commits us to the position that systems are not merely conceptual abstractions but structured objects. The ontology does not address systems that cannot be represented as a graph, thereby implicitly constraining the domain to computer systems in which nodes are computationally capable entities (which can both physical and non-physical).

A second commitment concerns \textit{temporality}. Our ontology evaluates systems at a fixed point in time. The framework treats distinct operational phases as distinct topologies, where each can be assessed independently. This commitment entails that dynamic behaviour is handled through static "snapshots". The practical consequence is that a system undergoing phase-dependent structural change (such as a federated learning protocol alternating between a cryptographic setup phase and an inference phase \cite{secaggplus}) is not treated as a single topology that evolves, but as a family of topologies corresponding to distinct points in time, each governed by its own constraint set. 
% it would be good to argue this further in the discussion

Lastly, the ontology adopts a \textit{multi-dimensionality} for subjects of decentralization. Following from what has been extensively discussed thus far,  we recognize that any contextually defined component of a system (whether it would be data localization, model training, etc.) constitutes a legitimate \textbf{subject of decentralization}, provided that it admits a graph-realistic representation. The interpretation against which center multiplicity is assessed is kept separate as an \textit{anchor}, such as ownership, authority, or trust. Thus, “data decentralized with respect to ownership” is an evaluation of a subject–anchor pair rather than an equation of data with ownership. The number of such pairs classified as decentralized defines the system’s
dimensionality relative to the declared evaluation profile, with a centralized system constituting the degenerate zero-dimensional case.

No fixed inventory or cardinality is imposed on subjects or anchors. Their interpretations are supplied for the system and context being modeled, while the center counting rule remains the same across instantiations.

% Revisit the description language later to make it more precise to what we actually use from it

% we define our ontology using $\mathcal{SROIQ(D)}$ description logic (which we denote as $\mathcal{L}$) as our baseline, and appeal to predicate logic where needed (taking a similar approach to that of \cite{gnatenkol_building_2024}).

Now formally, we follow the approach presented in \cite{guarino1998} for consistency with established ontology literature, as such, given a formal language $\mathcal{L}$, we construct our ontological commitment as follows:
\begin{equation}
    \mathcal{K} =\ \textless \mathcal{C},\, \mathcal{J} \textgreater
\end{equation}
Where $\mathcal{C}$ is our formal conceptualization expressed as:
\begin{equation}
    \mathcal{C} =\ \textless \Delta,\, \mathcal{W},\, \mathcal{R} \textgreater
\end{equation}
Where $\Delta$ is our domain (containing entities relevant to the computational systems under consideration), $\mathcal{W}$ is a set of possible worlds (maximal states of affairs for our domain \cite{guarino1998}, in our case, arrangement of these entities and their subjects of decentralization in the system), and $\mathcal{R}$ represents a set of conceptual relations for $\textless \Delta,\, \mathcal{W}\textgreater$. Meanwhile $\mathcal{J}$ denotes a conceptual interpretation function that maps elements from our signature $\Sigma$ (vocabulary, which will be specified in the following subsections) to their intended meanings within the conceptualization $\mathcal{C}$.

A model (or a model structure) for the language $\mathcal{L}$ can be expressed as $\textless \mathcal{M},\, \mathrm{I}\textgreater$ such that $\mathcal{M}$ denotes a world-relative relational structure, while $\mathrm{I}$ provides the interpretation mapping over that structure. $\mathcal{M}$ can be expressed as:
\begin{equation}
    \mathcal{M} = \textless \Delta,\, \mathrm{R}\textgreater
\end{equation}
Where $\mathrm{R}$ is a set of extensions relative to the world specific structure that models $\mathcal{L}$ such that $\mathrm{R} = \{ p(w)\ |\ p \in \mathcal{R},\, w \in \mathcal{W} \}$. Note that $\mathcal{R}$ denotes conceptual relations at the level of conceptualization, whereas $\mathrm{R}$ denotes their world-relative extensions. 

%Altogether using interpretation mapping:
%\begin{equation}
%    \mathrm{I}=(\Delta^I,\, \cdot^I)\quad \text{where}\quad \cdot^I : \Sigma \rightarrow \Delta^I\, \cup\, \mathrm{R} \label{eq:interpretation_function}
%\end{equation}
%Notice how Equation \ref{eq:interpretation_function} simply describes interpretation domain $\Delta^I$ with the associated interpretation function $\cdot^I$ that [[maps accordingly]] to our ontological commitment $\mathcal{K}$. 

With our commitment we construct an ontology $\mathcal{O}$ for the language $\mathcal{L}$ that commits to conceptualization $\mathcal{C}$:
\begin{equation}
    \mathcal{O} = \textless \Delta,\, \mathcal{W},\, \mathcal{R},\, \mathcal{J},\, \Sigma,\, \Phi \textgreater 
\end{equation}

A model $\textless \mathcal{M},\, \mathrm{I}\textgreater$ satisfies our ontology $\mathcal{O}$ as long as it satisfies our axioms $\varphi$, this follows the standard textbook notation \cite{baader_description_2007}.
\begin{equation}
        \textless \mathcal{M},\, \mathrm{I}\textgreater \vDash \mathcal{O} \iff \forall \varphi \in \Phi,\; \textless \mathcal{M},\, \mathrm{I}\textgreater \vDash \varphi \label{model_ontology}
\end{equation}

%Where $\mathcal{T}$ and $\mathcal{A}$ are our knowledge base with TBox $\mathcal{T}$ consisting of terminological axioms, ABox $\mathcal{A}$ containing assertions, and $\Sigma$ representing the ontology signature. We refer to \cite{baader_description_2007} for textbook definitions.

% Address concretely the abstraction level (besides just saying that "oh its an ontology")

\vspace{2mm}

\noindent\textbf{Ontology Signature.} In principle, we aim for our ontology to be as minimal as it is possible while sustaining the necessary expressiveness to define decentralized systems. Reason being is to avoid over-axiomatization with overly large number of primitive concepts. With that said, we define our ontology singature in the following manner.
\begin{equation}
    \Sigma = (C, R)
\end{equation}
Where $C$ denotes our primitive concepts and $R$ primitive relations. Our primitive concepts consist of the following:
\begin{equation}
    C = \{ System,\, Topology,\, Subject,\, Anchor \}
\end{equation}
These terms are defined as follows: $System$ is a set of computational entities that admits a structural representation, $Topology$ is the said graph-realistic relational structure representation, $Subject$ is a context-relative projection over the structure representation, and $Anchor$ represents the interpretation under which centers of centralization are distinguished. 
Similarly, our primitive relations contain:
\begin{equation}
    \begin{aligned}
        R = \{\, & hasTopology,\, hasSubject,\, hasProjection,\\
                 & ofSubject,\, inTopology,\, hasAnchor,\ \\
                 & hasRealization,\, realizedAt\, \}
    \end{aligned}
\end{equation}
Where $hasTopology$ indicates a relation between a given graph-based structure (which can be expressed with graph-theoretic methods; e.g., adjacency matrices) and a system instance, while $hasSubject$ describes a relation between a system instance and a context-relative subject. The relation $hasProjection$ connects a system to a projection particular $p$; $ofSubject$ and $inTopology$ identify, respectively, the unique subject and topology associated with $p$; $hasAnchor$ connects a given system with an anchor interpretation $a$, $hasRealization$ relates $p$ to its distinct realization particulars; and $realizedAt$ links each realization particular to the vertex at which it is realized. Projection and realization particulars are not introduced as additional primitive concepts as they are individuals in $\Delta$ identified through these relation patterns. Likewise, center interpretations are represented extensionally as subsets of the topology vertex set, so no separate center role is required.

\vspace{2mm}

\noindent\textbf{Ontology Taxonomy and Schema.} Next, we establish a taxonomy as its proper part \cite{arp_building_ontologies} where a hierarchy is required as it follows a graph-theoretic structure with a single root node. 

However, we are interested in addressing the decentralization problem through a clear separation of representational units that altogether constitute a system (specifically, those that can characterize a system particular as canonically decentralized or not). This must be done without constructing our taxonomy backwards from the terms of decentralization or distribution, as that would risk cyclic definitions. This way we also avoid binding ourselves to a specific context in which decentralization could be used in. 

Providing the context-dependent expressions of $Subject$ and $Topology$, it isn't possible to construct an \textit{is\_a} hierarchy between them (as well as the $System$) as they are described via other relations specified in $R$. Instead, each concept in $C$ has a respective taxonomic backbone. In order to sustain perspectivalism and realism of our ontology, both $Topology$, $Subject$, and $Anchor$ taxonomies can be expanded upon and established with specific contexts during the design and/or evaluation of a system. 

\begin{figure}[ht]
    \centering
    \includegraphics[width=0.3\textwidth]{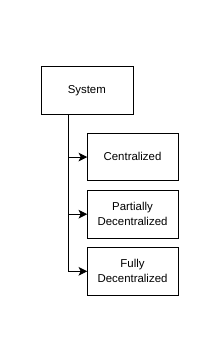}
    \caption{System-class taxonomic backbone.}
    \label{fig:taxonomical_backbone}
\end{figure}

Meanwhile, the set of representational units subsumed by $System$ concept is finite (Figure \ref{fig:taxonomical_backbone}), each describing the extent of decentralization a system has that follows from our system-centric ontological schema (Figure \ref{fig:ontology_schema}). The schema shows the mapping of non-taxonomical relations from $R$ across all concepts in $C$, along with additional relations describing how particulars are instantiated (for $Topology$).  

\begin{figure}[ht]
    \centering
    \includegraphics[width=0.5\textwidth]{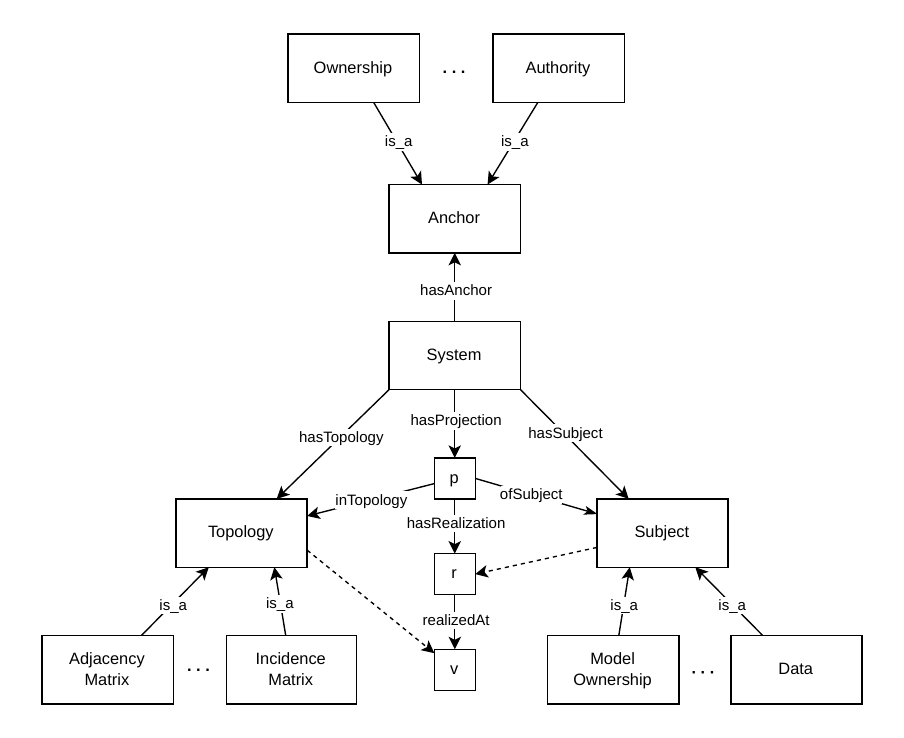}
    \caption{System-centric ontology schema. Individuals p, r, and v denote representative instances (particulars) used to illustrate object property assertions between ontology classes. Dashed arrows indicate informal conceptual associations only; v is obtained from the topology representation and r represents a realization of the subject, included solely to aid interpretation. Solid arrows denote ontology roles.}
    \label{fig:ontology_schema}
\end{figure}

\vspace{2mm}

% talk about perspectivalism (tehe are multiple accurate descriptions of reality), as that is one of the general principles

\textbf{More on Projection Particulars and Multiplicity.} The last caveat that needs addressing is the intention and interpretation of projection particulars. Our ontology identifies a projection particular $p$ through its relational structure. Under an interpretation $\mathrm{I}$, the set of projection particulars is derived as:
\begin{equation*}
\begin{aligned}
    \Pi^{I}
    = \bigl\{p\in\Delta^{I}\ \bigm|\ & \exists s,u,T\,\bigl(hasProjection(s,p) \\
    & \wedge ofSubject(p,u) \wedge inTopology(p,T) \bigr) \bigr\}
\end{aligned}
\end{equation*}
For a subject $u$ belonging to a system $s$ with topology $T$, we write $p_u$ for the unique projection particular satisfying $hasProjection(s,p_u)$, $ofSubject(p_u,u)$, and $inTopology(p_u,T)$. The multiplicity of the subject in that context is then derived from the distinct realization particulars related to $p_u$. Intuitively, $p$ can be thought of as a projection linking the system, topology, and subject together.

This distinction is necessary because realization multiplicity, distribution, and decentralization are not equivalent. We later derive $\delta(p_u)$ as the count of realization particulars and $\lambda(p_u)$ as the count of distinct vertices from a given topology supporting them. Multiple realization particulars may therefore be located at one vertex, so $\delta(p_u)>1$ with $\lambda(p_u)=1$ is entirely admissible. Crucially however, count does not by itself states how many ownership, authority, trust, or other effective centers those vertices represent. For that purpose, $\mu(p_u,a)$ counts the center ``regions'' reached by the subject under anchor $a$. Briefly, distribution depends on $\lambda$, anchor relative centralization and decentralization depends on $\mu$, and $\delta$ remains the raw multiplicity from which support is derived.

Notice how this interpretation is substantially different from works like \cite{Troncoso_2017} which focus on trust and direct control, but often limit the discussion on what can and cannot be controlled (e.g., an entity holds data which belongs to multiple clients  is cryptographically bidden, that same entity can erase data completely from itself but cannot modify it without clients noticing) as well as particular contexts (blockchain-based privacy). This is partly the reason why our focus shifts directly onto subjects themselves, as we want to avoid restricting ourselves to particular elements that are associated with decentralization in some contexts. 

\par

\vspace{2mm}

\noindent\textbf{Competency Questions.} The questions this ontology answers address the existence and extent of decentralization across different dimensions in reference to the graph-realistic structure any system can be described as. It allows to identify the exact components that are to be decentralized (subjects of decentralization) in accordance with the intended function of the system (contextual operationalism), as well as the topological relations defining how those components are structurally positioned and interconnected (structural relationalism). Take note that, in principle, competency questions themselves are not a part of our ontology, but are here to present external constraints for validation. With that said, to construct our competency questions we follow the methodology from \cite{gruningerfox_method}, where we first phrase them informally to then develop their formal instances in later sections:
\begin{itemize}
    % \item \textbf{CQ 1.} Which system components can be decentralized? (search for "self note" to see why are we removing this)
    \item \textbf{CQ 1.} When is a component considered to be centralized or decentralized?
    \item \textbf{CQ 2.} When is a component considered to be distributed?
    \item \textbf{CQ 3.} How many dimensions of decentralization does a system have?
    \item \textbf{CQ 4.} How do we derive a notion for centralized system from decentralization? 
    \item \textbf{CQ 5.} Can we state that a system is fully decentralized?
    \item \textbf{CQ 6.} Can a system be decentralized in one component but centralized in another?
    \item \textbf{CQ 7.} Can a distributed subject remain centralized with respect to an anchor?
\end{itemize}

Altogether, above competency questions make up for external constraints on acceptable models of the ontology.

\section{Logical Layer}

\subsection{Formalisms}

\noindent\textbf{Semantics.} We approach semantics for our ontology using competency questions and formalize their semantic expressions. Take note that we carry over the symbol notation throughout our semantics due to space limits. Let $U_s$ denote all subjects linked to $s$, and let $\mathcal{E}_s\subseteq U_s\times\mathcal{A}_s$ be the finite evaluation profile chosen for that system taking system subjects and system anchors of centralization $\mathcal{A}_s$. A pair $(u,a)\in\mathcal{E}_s$ means that subject $u$ is evaluated with respect to anchor $a$; it does not introduce another primitive role. The center-counting function $\mu(p_u,a)$, as well as other functions, are derived below.

\vspace{2mm}

\noindent\textit{\textbf{CQ 1.}} A subject is centralized or decentralized only relative to an anchor:
\begin{equation*}
    \begin{aligned}
        \mu(p_u,a)=1
            &\iff \textsc{Centralized}(u,a),\\
        \mu(p_u,a)>1
            &\iff \textsc{Decentralized}(u,a),\\
        \text{where:}\quad &(u,a)\in\mathcal{E}_s
        \wedge hasProjection(s,p_u)\\
        &{}\wedge ofSubject(p_u,u)\wedge inTopology(p_u,T)
    \end{aligned}
\end{equation*}
The phrase ``$u$ is decentralized/centralized'' is therefore shorthand to the anchor and evaluation profile $\mathcal{E}_s$, this must be supplied by context.  

\noindent\textit{\textbf{CQ 2.}} The notion of distribution in our ontology is defined by the subject realizations over vertices. We say that a subject of decentralization is \textit{distributed} if and only if the realizations associated with its projection particular are located over more than one vertex.
    \begin{equation*}
        \lambda(p_u) > 1 \iff \textsc{Distributed}(u), \quad \lambda(p_u)=|V_{p_u}|
    \end{equation*}
Where $V_{p_u}\subseteq V$ is a derived set of vertices supporting at least one realization associated with $p_u$ and the realization of $T$ is a graph $G(V,E)$. Making $\lambda$ independent of the interpretation given by an anchor particular $a$. 

\noindent\textit{\textbf{CQ 3.}} Dimensionality of decentralization in a given model is defined by the count of subjects under which it is decentralized:
\begin{equation*}
    \begin{aligned}
        D(s) & = \{ (u, a) \in \mathcal{E}_s\ |\ \mu(p_u, a) > 1 \}, \\
        dim(s, \mathcal{E}_s)&=|D(s)| \\
    \end{aligned}
\end{equation*}

\noindent\textit{\textbf{CQ 4.}} Centralization of a whole model is implied directly by the absence of decentralized subjects w.r.t the evaluation profile $\mathcal{E}_s$, meaning that the dimension of decentralization for a centralized system equals 0:
\begin{equation*}
    \begin{aligned}
        \textsc{Centralized}(s, \mathcal{E}_s) \iff\ & System(s)\ \wedge \\
        & \forall (u, a) \in \mathcal{E}_s\ (\mu(p_u, a) = 1), \\
        \textsc{Centralized}(s, \mathcal{E}_s) \implies\ & dim(s)=0
    \end{aligned}
\end{equation*}

\noindent\textit{\textbf{CQ 5.}} Now we semantically introduce a formal distinction between a system that has ``some'' decentralization and one that is fully decentralized. Full decentralization requires at least one subject projection and for that projection to be realized at more than one center for every anchor of centralization:
\begin{equation*}
    \begin{aligned}
        \textsc{Fully Decentralized}(s, \mathcal{E}_s) \iff \\
        & \hspace{-30pt} System(s)\ \wedge \\
        & \hspace{-30pt} \forall (u, a) \in \mathcal{E}_s\ (\mu(p_u, a) > 1)
    \end{aligned}
\end{equation*}
    
Similarly, a model is partially decentralized when it contains at least one centralized subject-anchor pair that is centralized and one that is decentralized:
\begin{equation*}
    \begin{aligned}
        \textsc{Partially Decentralized}(s, \mathcal{E}_u) \iff \\
        & \hspace{-70pt} System(s)\ \wedge \\
        & \hspace{-70pt} \exists (u,a) \in \mathcal{E}_s\ (\mu(p_u,a) = 1)\ \wedge \\
        & \hspace{-70pt} \exists (u,a) \in \mathcal{E}_s\ (\mu(p_u,a) > 1)
    \end{aligned}
\end{equation*}

\noindent\textit{\textbf{CQ 6.}} The partial case shows that distinct subject-anchor pairs in one system may receive different classifications. The same subject may also be evaluated under two separate anchors.

\noindent\textit{\textbf{CQ 7.}} A distributed subject may be realized across vertices that belong to a singular center in a given anchor. As such the following is satisfiable:

\begin{equation*}
    \begin{aligned}
        \lambda(p_u) > 1\ \wedge\ \mu(p_u, a) = 1
    \end{aligned}
\end{equation*}

Here, distribution still entails that realization must be greater than one, but ensures that the multiplicity itself $\delta(p_u) > 1$ is not within of itself a definition of decentralization.

\par

\vspace{2mm}

% We should also justify our modelling choices- most likely in axiomitizaiton?

\noindent\textbf{Axiomitization.} For clarity, axioms presented here serve as a ``specification'' by encoding formal constraints over our ontology's intended models in accordance with \cite{gruningerfox_method}. Where our semantics specify what it means for a subject to be centralized or for a system to be fully decentralized, and competency questions clarify the purpose and exact questions our ontology is meant to answer; in this section we specify what must hold in any model satisfying $\mathcal{O}$, ruling out interpretations that are formally consistent with $\Sigma$ but ontologically unintended. We use description-logic TBox statements for role typing and cardinality restrictions, supported by first-order constraints where relation between several role paths must be expressed. We follow the standard treatment of \cite{baader_description_2007} for TBox construction.\par

\noindent\textbf{Domain and Range.} We define every relation in $R$ through formal domain, range, and relational-structure constraints:
\begin{equation*}
    \begin{array}{l}
        \top \sqsubseteq \forall hasTopology^-.System \\
        \top \sqsubseteq \forall hasTopology.Topology \\
        \\[-2pt]
        \top \sqsubseteq \forall hasSubject^-.System \\
        \top \sqsubseteq \forall hasSubject.Subject \\
        \\[-2pt]
        \top \sqsubseteq \forall hasAnchor^-.System \\
        \top \sqsubseteq \forall hasAnchor.Anchor \\
        \\[-2pt]
        \top \sqsubseteq \forall hasProjection^-.System \\
        \top \sqsubseteq \forall ofSubject.Subject \\
        \top \sqsubseteq \forall inTopology.Topology \\
        \\[-2pt]
        \exists ofSubject.\top
            \sqsubseteq \exists hasProjection^-.System \\
        \exists inTopology.\top
            \sqsubseteq \exists hasProjection^-.System \\
        \exists hasRealization.\top
            \sqsubseteq \exists hasProjection^-.System \\
        \exists realizedAt.\top
            \sqsubseteq \exists hasRealization^-\top
    \end{array}
\end{equation*}
Since projection and realization particulars are not primitive concepts, their admissible extensions are derived from the relation patterns above rather than named as additional classes.

\noindent\textbf{Existence and Cardinality.} Existence of decentralizable subjects is not assumed by our ontology as a model with no subjects remains admissible and has dimensionality $0$. Consequently, no positive lower-bound axiom is imposed on $hasSubject$. Every system must, however, have exactly one topology:
\begin{equation*}
    System \sqsubseteq (=1\,hasTopology.Topology)
\end{equation*}
Because subjects are context-relative particulars, every subject belongs to exactly one system:
\begin{equation*}
    Subject \sqsubseteq (=1\,hasSubject^-.System)
\end{equation*}

Anchor particulars are also system-relative:

\begin{equation}
    Anchor \sqsubseteq (=1\, hasAnchor^1.Sysytem)
\end{equation}

With no lower bound imposed on $hasAnchor$ at the ontological layer. Every projection particular belongs to exactly one system and, within that system, has exactly one subject, exactly one topology, and at least one realization particular:
\begin{equation*}
    \begin{aligned}
        \exists hasProjection^-.\top
        &\sqsubseteq
        (=1\,hasProjection^-.System),
        \\
        System
        &\sqsubseteq
        \forall hasProjection.\bigl(
        \\
        &\quad (=1\,ofSubject.Subject)
        \\
        &\quad{}\sqcap (=1\,inTopology.Topology)
        \\
        &\quad{}\sqcap (\geq1\,hasRealization.\top)
        \bigr)
    \end{aligned}
\end{equation*}
Every realization particular belongs to exactly one projection particular and is located at exactly one vertex:
    \begin{equation*}
        \begin{aligned}
            \exists hasRealization^-.\top
            \sqsubseteq {}&
            (=1\,hasRealization^-.\top)
            \\
            &{}\sqcap (=1\,realizedAt.\top)
        \end{aligned}
    \end{equation*}
The following constraints ensure that the role fillers refer to the same system context and that every subject admitted by a system has exactly one projection particular:
    \begin{equation*}
        \begin{aligned}
            \forall s,u\,\bigl(& hasSubject(s,u) \\
            & \rightarrow\exists!p\,\bigl(hasProjection(s,p) \\
            & \hspace{31mm}{}\wedge ofSubject(p,u) \bigr)\bigr)
        \end{aligned}
    \end{equation*}

    \begin{equation*}
        \begin{aligned}
            \forall s,p\,\bigl(& hasProjection(s,p) \\
            & \rightarrow \exists!u\,\bigl(hasSubject(s,u) \wedge ofSubject(p,u) \bigr) \\
            & \phantom{\rightarrow}\wedge \exists!T\,\bigl(hasTopology(s,T)\wedge inTopology(p,T) \bigr)\bigr)
        \end{aligned}
    \end{equation*}

Finally, realization locations (vertices) must belong to the graph that realizes the topology associated with the same projection:
\begin{equation*}
    \begin{aligned}
        \forall p,T,r,v\,\bigl(& inTopology(p,T) \wedge hasRealization(p,r) \\
        & \wedge realizedAt(r,v) \\
        & \rightarrow v\in V_T \bigr)
    \end{aligned}
\end{equation*}
Where $V_T$ denotes the vertex set of the graph realizing $T$.

\noindent \textbf{Multiplicity, Placement, Anchors, and centers.} For each projection particular $p\in\Pi^I$, define a finite and non-empty realization set:
\begin{equation*}
    \mathcal{R}_{p}=\{r\in\Delta^I\mid(p,r)\in hasRealization^I\}
\end{equation*}
And the raw multiplicity function:
\begin{equation*}
    \delta:\Pi^I\longrightarrow\mathbb{N}_{>0},
    \qquad \delta(p)=|\mathcal{R}_{p}|
\end{equation*}
Thus $\delta$ counts solely realization particulars. The supporting vertices and their cardinality are derived separately:
\begin{align*}
    V_p & =\{v\in V_T\mid\exists r\in\mathcal{R}_p\,
        ((r,v)\in realizedAt^I)\},\\
    \lambda(p) & = |V_p|,
    \qquad \lambda:\Pi^I\longrightarrow\mathbb{N}_{>0}
\end{align*}
Because each realization has exactly one location while several realizations may share one location:
\begin{equation*}
    1\leq\lambda(p)\leq\delta(p)
\end{equation*}
Consequently, $\delta(p)>1$ with $\lambda(p)=1$ is satisfiable, and $\delta(p)=1$ entails $\lambda(p)=1$. These are claims about raw multiplicity and placement, not about decentralization.

For a system $s$, the set of anchor interpretations is:
\begin{align*}
    \mathcal{A}_s
    &=\{\textit{Ownership},\textit{Authority},\ldots\}\\
    &=\{a\in Anchor^I\mid(s,a)\in hasAnchor^I\}
\end{align*}
For each $a\in\mathcal{A}_s$, its centers are interpreted extensionally as a subset of vertices:
\begin{align*}
    \mathcal{C}_a
    &=\{\textit{Authority}_1,\textit{Authority}_2,\ldots\}\\
    &=\{V_{a,1},V_{a,2},\ldots,V_{a,n}\}
      \subseteq\mathcal{P}(V_T)
\end{align*}
Formally, each member of $\mathcal{C}_a$ is already a subset of $V_T$, so a separate $V_c$ definition or center role would be redundant.
An admissible evaluation supplies the finite mapping $a\mapsto\mathcal{C}_a$ for every anchor used in $\mathcal{E}_s$ as part of the graph-level interpretation, not as another primitive relation in $R$.

If no shared or external center is declared, the default family is:
\begin{equation*}
    \mathcal{C}_a^{\mathrm{default}}
    =\{\{v\}\mid v\in V_T\}\subseteq\mathcal{P}(V_T)
\end{equation*}
This also applies in cases where there exist externally defined centers but they are not covering the whole topology, the default case applies to all remaining vertices outside of these centers. Following an intuition that each vertex would function as its own internal center for that anchor (e.g., ownership declared for some nodes but not all, each remaining node in that case is an owner of itself). 

As such, if declared regions cover only part of the topology, every uncovered vertex is added as an independent singleton effective center. Henceforth, $\mathcal{C}_a$ denotes this completed family, so:
\begin{equation*}
    V_{p_u}\subseteq
    \bigcup_{V_c\in\mathcal{C}_a}V_c=V_T
\end{equation*}

The anchor-relative center multiplicity of a subject projection is then:
\begin{equation*}
    \begin{aligned}
        \mu(p_u,a)
        & = \bigl|\{V_c\in\mathcal{C}_a\mid V_c\cap V_{p_u}\neq\varnothing\}\bigr|,\\[5pt]
        \hspace{5pt}V_c\cap V_{p_u}\neq\varnothing
        & \iff\exists r\,\exists v\,\bigl(
            hasRealization(p_u,r)\\
        &\hspace{28pt}\wedge realizedAt(r,v)\wedge v\in V_c\bigr)
    \end{aligned}
\end{equation*}
We use existential quantifiers to state that a center is counted when at least one realization of $p_u$ occurs at at least one vertex in that region.

\vspace{2mm}

\noindent\textbf{Refined Concepts.} For every declared pair $(u,a)\in\mathcal{E}_s$, completed coverage, and non-empty subject support ensure $\mu(p_u,a)\in\mathbb{N}_{>0}$. We then infer:
\begin{align*}
    \textsc{Centralized}(u,a)
        & \iff Subject(u)\wedge\mu(p_u,a)=1,\\
    \textsc{Decentralized}(u,a)
        & \iff Subject(u)\wedge\mu(p_u,a)>1,\\
    \textsc{Distributed}(u)
        & \iff Subject(u)\wedge\lambda(p_u)>1
\end{align*}
For a fixed evaluation profile $\mathcal{E}_s$, each declared pair is classified exhaustively and exclusively, where a single subject cannot be classified as both centralized and decentralized simultaneously w.r.t the same anchor.  System-level \textsc{Centralized}, \textsc{Partially Decentralized}, and \textsc{Fully Decentralized} classifications are the profile-relative predicates defined under CQ 4 and CQ 5. Comparisons are valid only when the declared evaluation profiles are compatible.

\vspace{2mm}

\noindent\textbf{Reasoning Properties.} We validate the reasoning properties of our ontology using standard approaches by showing: consistency, satisfiability, disjointness, and entailment.

\vspace{2mm}

\noindent\textbf{Definition 1. (\textit{Consistency}).}
\textit{Let $\mathcal{O}$ be an ontology with a set of axioms $\Phi$. The ontology
$\mathcal{O}$ is said to be \emph{consistent} if and only if there exists a
structure $\mathcal{M}$ and an interpretation $\mathrm{I}$ such that
$\mathcal{M}$ under $\mathrm{I}$ satisfies every axiom in $\Phi$. Formally:}
\begin{equation*}
\begin{aligned}
    \operatorname{Consistent}(\mathcal{O})
    &\iff
    \exists \textless \mathcal{M},\, \mathrm{I}\textgreater
    \text{ such that }
    \textless \mathcal{M},\, \mathrm{I}\textgreater
    \vDash \mathcal{O}
    \\
    &\iff
    \exists \textless \mathcal{M},\, \mathrm{I}\textgreater
    \text{ such that }
    \forall \varphi \in \Phi,
    \\
    &\hspace{22mm}
    \textless \mathcal{M},\, \mathrm{I}\textgreater
    \vDash \varphi
\end{aligned}
\end{equation*}

\vspace{2mm}

\noindent \textbf{\textit{Proof.}} To prove consistency we are required to show that there exists an interpretation that models our ontology:
    \begin{equation*}
        \exists \textless \mathcal{M},\, \mathrm{I}\textgreater
        \vDash \mathcal{O}
    \end{equation*}
Construct an interpretation $\mathrm{I}=\textless\Delta^I,\cdot^I\textgreater$ and a witness structure $\mathfrak{W}=\{s,t,U,A,P,Q,V\}$, where $s\in System^I$, $t\in Topology^I$ realized by a non-empty finite graph $G(V,E)$, $U=\{u_1,\ldots,u_n\}\subseteq Subject^I$, $A=\{a_1,\ldots,a_k\}\subseteq Anchor^I$, $P=\{p_1,\ldots,p_n\}\subseteq\Pi^I$, and $Q=\bigcup_i\mathcal{R}_{p_i}$. For every $u_i$, let $p_i$ be its unique projection particular, let $\mathcal{R}_{p_i}$ be finite and non-empty, and let $\ell:Q\rightarrow V$ assign each realization to exactly one vertex in this particular case. For every $a_j$, choose a completed finite family $\mathcal{C}_{a_j}\subseteq\mathcal{P}(V)$ whose union is $V$, and choose a finite profile $\mathcal{E}_s\subseteq U\times A$. Role interpretations are:
    \begin{equation*}
        \begin{array}{l}
             hasTopology^I = \{(s,t)\}, \\
             hasSubject^I = \{(s,u_i) : u_i \in U\}, \\
             hasAnchor^I = \{(s,a_j) : a_j \in A\}, \\
             hasProjection^I = \{(s,p_i) : p_i \in P\}, \\
             ofSubject^I = \{(p_i,u_i) : 1\leq i\leq n\}, \\
             inTopology^I = \{(p_i,t) : p_i \in P\}, \\
             hasRealization^I =
             \{(p_i,r) : p_i\in P,\ r\in\mathcal{R}_{p_i}\}, \\
             realizedAt^I = \{(r,\ell(r)) : r\in Q\}
        \end{array}
    \end{equation*}
Under this interpretation, every subject and anchor belongs to one system, every projection has one system, subject, and topology and at least one finite realization, and every realization has one topological location. Completion of each center family makes $\mu$ a positive integer for every declared pair. All axioms in $\Phi$ are therefore satisfied:
    \begin{equation*}
        \textless \mathfrak{W},\,\mathrm{I}\textgreater
        \vDash \mathcal{O}
    \end{equation*}
Hence, the ontology is consistent.
\hfill\(\square\)

\vspace{2mm}

\noindent\textbf{Definition 2. (\textit{Disjointness})} \textit{Let $\mathcal{O}$ be an ontology, and let $C$ and $D$ be concepts in the signature of $\mathcal{O}$. The concepts $C$ and $D$ are said to be \emph{disjoint} with respect to $\mathcal{O}$ if and only if their interpretations have no common instances in every model of $\mathcal{O}$. Formally:}
\begin{equation*} 
    \begin{aligned} 
        \operatorname{Disjoint}_{\mathcal{O}}(C,D) &\iff \mathcal{O} \vDash C \sqcap D \sqsubseteq \bot \\ 
        &   \iff \forall \textless \mathcal{M},\, \mathrm{I}\textgreater, \\ 
        &   \hspace{10mm} \textless \mathcal{M},\, \mathrm{I}\textgreater \vDash \mathcal{O} \Longrightarrow C^I \cap D^I = \varnothing.
    \end{aligned} 
\end{equation*}

\vspace{2mm}

\noindent\textbf{\textit{Proof.}} Fix a system $s$ and evaluation profile $\mathcal{E}_s$. For each $(u,a)\in\mathcal{E}_s$, complete center coverage (guaranteed by the default $C_a^{default}$), and non-empty subject support yield $\mu(p_u,a)\in\mathbb{N}_{>0}$. Exactly one of $\mu(p_u,a)=1$ and $\mu(p_u,a)>1$ can holds, therefore:
\begin{equation*}
    \textsc{Centralized}(u,a)\wedge
    \textsc{Decentralized}(u,a)\rightarrow\bot.
\end{equation*}
The profile-relative system concepts are disjoint for the same reason given that a non-empty profile cannot simultaneously have every entry above one, every entry equal to one, or a mixture of the two.

Crucially, \textsc{Distributed} and anchor-relative \textsc{Centralized} are \emph{not} disjoint. Let $V_{p_u}=\{v_1,v_2\}$ and $\mathcal{C}_a=\{V_T\}$. Then $\lambda(p_u)=2$ while $\mu(p_u,a)=1$, so:
\begin{equation*}
    \textsc{Distributed}(u)\wedge
    \textsc{Centralized}(u,a)
\end{equation*}
Is satisfiable. If instead the completed family contains the two singleton regions $\{v_1\}$ and $\{v_2\}$, then $\lambda(p_u)=\mu(p_u,a)=2$ and $\textsc{Distributed}(u)\wedge\textsc{Decentralized}(u,a)$ is satisfiable.

Hence, the required class disjointness holds without making distribution and centralization mutually exclusive.
\hfill\(\square\)

\vspace{2mm}

\noindent\textbf{Definition 3. (\textit{Satisfiability})} \textit{Let $\mathcal{O}$ be an ontology, and let $C$ be a concept in the signature of $\mathcal{O}$. The ontology $\mathcal{O}$ is said to be \emph{satisfiable} if and only if it admits at least one model. The concept $C$ is said to be \emph{satisfiable} with respect to $\mathcal{O}$ if and only if its interpretation is non-empty in at least one model of $\mathcal{O}$. Formally:}
\begin{equation*}
    \begin{aligned}
        \operatorname{Satisfiable}(\mathcal{O}) & \iff \exists \textless \mathcal{M},\, \mathrm{I}\textgreater, \\
        & \hspace{10mm} \textless \mathcal{M},\, \mathrm{I}\textgreater
        \vDash \mathcal{O}, \\
        \operatorname{Satisfiable}_{\mathcal{O}}(C) & \iff \mathcal{O} \not\vDash C \sqsubseteq \bot \\
        & \iff \exists \textless \mathcal{M},\, \mathrm{I}\textgreater, \\
        & \hspace{10mm} \textless \mathcal{M},\, \mathrm{I}\textgreater \vDash \mathcal{O} \quad \text{and} \quad C^I \neq \varnothing
    \end{aligned}
\end{equation*}

\noindent\textbf{\textit{Proof.}} The first expression is equivalent to consistency and therefore holds for our ontology. In the witness model $\mathfrak{W}$, all domain, range, existence, and cardinality constraints are satisfied. Each projection has a finite non-empty realization set and each realization has exactly one location, so $\delta(p_i)\geq1$ and $\lambda(p_i)\geq1$. Each completed center family covers $V_T$, so $\mu(p_i,a_j)\geq1$ whenever $(u_i,a_j)\in\mathcal{E}_s$.

Both subject-specific classifications are satisfiable. Choose $\mathcal{C}_{a_j}=\{V_T\}$ to obtain $\mu(p_i,a_j)=1$ and hence $\textsc{Centralized}(u_i,a_j)$. Choose a projection supported at $v_1$ and $v_2$ and use the singleton default family to obtain $\mu(p_i,a_j)=2$ and hence $\textsc{Decentralized}(u_i,a_j)$. The latter witness also satisfies $\textsc{Distributed}(u_i)$ because $\lambda(p_i)=2$; retaining the same two-vertex support but replacing the singleton family with whole of $V_T$ witnessing a distributed yet centralized subject w.r.t the anchor of centralization.

The three system-level concepts are satisfiable by choosing, respectively, a profile containing only $\mu=1$ entries, a non-empty profile containing only $\mu>1$ entries, or a profile containing both. An empty evaluation profile remains centralized by the vacuous universal condition and has $dim(s, \mathcal{E}_s)=0$. Therefore, all derived concepts are satisfiable within $\mathcal{O}$.
\hfill\(\square\)
    
\vspace{2mm}

\noindent\textbf{Definition 4. (\textit{Entailment})} \textit{Let $\mathcal{O}$ be an ontology, and let $\alpha$ be an axiom or assertion expressed in the signature of $\mathcal{O}$. The ontology $\mathcal{O}$ is said to \emph{entail} $\alpha$ if and only if every model of $\mathcal{O}$ is also a model of $\alpha$. Formally:}
\begin{equation*}
    \begin{aligned}
        \mathcal{O} \vDash \alpha
        &\iff
        \forall \textless \mathcal{M},\, \mathrm{I}\textgreater,
        \\
        &\hspace{10mm}
        \textless \mathcal{M},\, \mathrm{I}\textgreater
        \vDash \mathcal{O}
        \Longrightarrow
        \textless \mathcal{M},\, \mathrm{I}\textgreater
        \vDash \alpha.
    \end{aligned}
\end{equation*}
\vspace{2mm}

\noindent\textbf{\textit{Proof.}} Let $\mathfrak{W} = \{s,t,U,A,P,Q,V\}$ satisfy $\mathcal{O}$, and fix $\mathcal{E}_s\subseteq U\times A$. For every declared pair, completed coverage entails $\mu(p_u,a)\geq1$. The definitions therefore entail exactly one pair-level classification according to whether $\mu(p_u,a)=1$ or $\mu(p_u,a)>1$. Applying the same exhaustive split to all entries of $\mathcal{E}_s$ entails the corresponding centralized, fully decentralized, or partially decentralized system classification and fixes $dim(s, \mathcal{E}_s)$ as the number of entries above one.

Distribution is entailed independently from $\lambda$:
\begin{equation*}
    \mathcal{O}\vDash
    (\textsc{Distributed}(u)\iff\lambda(p_u)>1 )
\end{equation*}
There is no entailment from this predicate to either anchor-relative class without $\mathcal{C}_a$. In particular, a single region containing every supporting vertex entails $\mu(p_u,a)=1$ regardless of how many such vertices exist. \hfill\(\square\)

\vspace{2mm}

\noindent\textbf{Evaluation.} We approach evaluation of our ontology twofold: (a) we evaluate our competency questions so that we can show how our ontology can actually provide answers to them, and (b) we provide instantiations of our ontology using architectural models such as Block-chain, and Federated Learning. 

\vspace{2mm}

\noindent\textbf{\textit{Competency Questions Evaluation.}} \textbf{CQ 1} is answered by $\mu(p_u,a)=1$ and $\mu(p_u,a)>1$, which are necessary and sufficient conditions for anchor-relative centralization and decentralization. \textbf{CQ 2} is answered independently by $\lambda(p_u)>1$, while $\delta(p_u)$ retains only its raw-multiplicity meaning.

\textbf{CQ 3} is answered by $dim(s,\mathcal{E}_s)=|\{(u,a)\in\mathcal{E}_s\mid\mu(p_u,a)>1\}|$. \textbf{CQ 4} follows because a profile with no decentralized entries has dimension zero and satisfies $\textsc{Centralized}(s,\mathcal{E}_s)$. \textbf{CQ 5} follows from the universal and mixed conditions over the same profile. These classifications are exhaustive and exclusive once $\mathcal{E}_s$ and the completed centers are fixed.

\textbf{CQ 6} is witnessed whenever one declared pair has $\mu=1$ and another has $\mu>1$, and \textbf{CQ 7} by a multi-vertex subject whose supporting vertices all lie within one center region: $\lambda(p_u)>1$ and $\mu(p_u,a)=1$.

\vspace{1mm}

\noindent\textbf{\textit{Instantiations.}} We first evaluate the examples ontologically and later reuse them in the analytical layer. Each example declares an evaluation profile and the center families needed for its classifications; $\delta$ and $\lambda$ are also reported to keep realization count and placement explicit.

\noindent\emph{Instantiation 1: Vanilla Federated Learning.} Consider a star topology in which $v_1$ is the aggregation server and $v_2$-$v_5$ are clients holding private data and performing local training. Let $a_{own}$ interpret ownership and $a_{auth}$ interpret authority, and use the singleton completed families $\mathcal{C}_{a_{own}}=\mathcal{C}_{a_{auth}}=\{\{v_i\}\mid1\leq i\leq5\}$. Meanwhile $u_{data}$, $u_{train}$, and $u_{Agg}$ denote data, training, and aggregation locations respectively.  We declare $\mathcal{E}_{FL}=\{(u_{data},a_{own}),(u_{train},a_{auth}),(u_{Agg},a_{auth})\}$:
\begin{itemize}
    \item $u_{\text{data}}$: $\delta(p_{data})=\lambda(p_{data})=4$ and $\mu(p_{data},a_{own})=4$, hence $\textsc{Decentralized}(u_{data},a_{own})$.
    \item $u_{\text{train}}$: $\delta(p_{train})=\lambda(p_{train})=4$ and $\mu(p_{train},a_{auth})=4$, hence $\textsc{Decentralized}(u_{train},a_{auth})$.
    \item $u_{\text{Agg}}$: $\delta(p_{Agg})=\lambda(p_{Agg})=1$ and $\mu(p_{Agg},a_{auth})=1$, hence $\textsc{Centralized}(u_{Agg},a_{auth})$.
\end{itemize}

Thus $dim(s_{FL},\mathcal{E}_{FL})=2$ and the system is $\textsc{Partially Decentralized}(s_{FL},\mathcal{E}_{FL})$: data ownership and training authority span several effective centers, while aggregation authority remains within one. This corresponds to conventional cross-silo federated learning in which data and local computation are spread across clients while global aggregation remains under one authority \cite{kairous2021_advances,chou_efficient_2021}. Data and training are also \textsc{Distributed}; aggregation authority is not.

\vspace{1mm}

\noindent\emph{Instantiation 2: Decentralized Federated Learning.} Retain the profile and singleton center families, but let all five clients hold data and train locally and let aggregation authority be realized at $v_1$, $v_2$, and $v_3$. Under this interpretation the three supporting vertices are distinct authority centers:
\begin{itemize}
    \item $u_{\text{data}}$: $\delta(p_{data})=\lambda(p_{data})=5$ and $\mu(p_{data},a_{own})=5$.
    \item $u_{\text{train}}$: $\delta(p_{train})=\lambda(p_{train})=5$ and $\mu(p_{train},a_{auth})=5$.
    \item $u_{\text{Agg}}$: $\delta(p_{Agg})=\lambda(p_{Agg})=3$ and $\mu(p_{Agg},a_{auth})=3$.
\end{itemize}

Every declared pair has $\mu>1$, so $dim(s_{FL},\mathcal{E}_{FL})=3$ and $\textsc{Fully Decentralized}(s_{FL},\mathcal{E}_{FL})$ holds. All three subjects are also \textsc{Distributed}. Note that this conclusion depends on the declared authority centers, not merely on the three aggregation locations. That is, if $\{v_1,v_2,v_3\}$ were one authority region, aggregation would remain distributed but centralized with respect to $a_{auth}$.

\vspace{1mm}

\noindent\emph{Instantiation 3: Blockchain 1.} Consider a permissioned blockchain over seven nodes. Node $v_2$ is the sole ordering authority, the four core nodes replicate the ledger, and $\{v_1,v_6,v_7\}$ submit transactions. Let $a_{auth}$ interpret authority, let $a_{own}$ interpret ownership, and declare $\mathcal{E}_{BC}=\{(u_{cons},a_{auth}),(u_{ledger},a_{own}),(u_{tx},a_{auth})\}$. With singleton effective centers for the stated independent entities:
\begin{itemize}
    \item $u_{\text{cons}}$: $\delta(p_{cons})=\lambda(p_{cons})=\mu(p_{cons},a_{auth})=1$, hence anchor-centralized.
    \item $u_{\text{ledger}}$: $\delta(p_{ledger})=\lambda(p_{ledger})=\mu(p_{ledger},a_{own})=4$, hence anchor-decentralized.
    \item $u_{\text{tx}}$: $\delta(p_{tx})=\lambda(p_{tx})=\mu(p_{tx},a_{auth})=3$, hence anchor-decentralized.
\end{itemize}

Accordingly, $dim(s_{BC},\mathcal{E}_{BC})=2$ and the system is partially decentralized relative to that profile. Ledger replication and transaction submission are distributed, while consensus authority is not.

\vspace{1mm}

\noindent\emph{Instantiation 4: Blockchain 2.} Retain the topology and profile, but let the four core nodes jointly validate and order blocks, every node replicate the ledger, and $v_3,v_4$ expose transaction submission. Under singleton effective centers:
\begin{itemize}
    \item $u_{\text{cons}}$: $\delta(p_{cons})=\lambda(p_{cons})=\mu(p_{cons},a_{auth})=4$.
    \item $u_{\text{ledger}}$: $\delta(p_{ledger})=\lambda(p_{ledger})=\mu(p_{ledger},a_{own})=7$.
    \item $u_{\text{tx}}$: $\delta(p_{tx})=\lambda(p_{tx})=\mu(p_{tx},a_{auth})=2$.
\end{itemize}

All declared pairs have $\mu>1$, so $dim(s_{BC},\mathcal{E}_{BC})=3$ and the system is fully decentralized relative to $\mathcal{E}_{BC}$; all three subjects are also distributed. Again, if the four validators share one authority center $V_c=\{v_2,v_3,v_4,v_5\}$, their unchanged values $\delta(p_{cons})=\lambda(p_{cons})=4$ coexist with $\mu(p_{cons},a_{auth})=1$. Consensus is then distributed but authority-centralized, and the system remains partially rather than fully decentralized under the same profile.

\vspace{1mm}

\noindent These paired instantiations demonstrate that decentralization is evaluated over subject-anchor pairs rather than assigned as an undifferentiated property of the system. Two systems with identical subject support may receive different anchor-relative classifications when their center interpretations differ. This is precisely the distinction that raw realization or vertex counts cannot recover.

\noindent Meanwhile, we can observe that the three functions have non-interchangeable roles. The value $\delta(p_u)$ counts distinct realization particulars, $\lambda(p_u)$ counts the vertices supporting them, and $\mu(p_u,a)$ counts the center regions reached under $a$. Thus $\delta(p_u)=4$ and $\lambda(p_u)=1$ describe several co-located realizations but do not, without an anchor interpretation, entail decentralization. Conversely, $\lambda(p_u)=4$ with $\mu(p_u,a)=1$ describes a distributed but centralized subject.

\noindent This captures deployment arrangements obscured by node counts alone. A cryptocurrency network may run several protocol nodes on distinct hosts that remain within one data center or under one operator \cite{gencer_2018_decentralization_bitcoin}. The hosts contribute to $\lambda$, while the declared ownership or authority region contributes only one to $\mu$. If no shared center is declared, the singleton completion convention instead treats every vertex as its own effective center.

\noindent The distinction also explains why decentralization admits partially and fully decentralized profile-relative system classes, whereas distribution is not assigned a partially distributed class. Following Sassoon's account of multidimensional adjectives, the relevant dimensions may be bound conjunctively, disjunctively, or contextually \cite{sassoon_typology_2013}. Here, our ontology records those dimensions explicitly in $\mathcal{E}_s$, and $dim(s,\mathcal{E}_s)$ counts the entries whose center multiplicity exceeds one.

\noindent For distribution, once the topology and its vertices have been fixed, it is the bivalent placement condition $\lambda(p_u)>1$ which decides whether a subject is distributed or not. There is no distribution analogue of $dim(s,\mathcal{E}_s)$ and therefore no intermediate \textsc{Partially Distributed} system class, reflecting common uses of ``distributed computing'' \cite{nakamoto_bitcoin_2008, distribution_reinforcement_abadi_2024, distributed_systems_textbook}. Decentralization is dimensional across declared subjects, whereas distribution records extensional topological placement.

\subsection{Definition for Decentralization}

Having constructed the ontology itself, we propose a definition for decentralization:

\vspace{2mm}

%\noindent\textbf{Definition 5. (\textit{Decentralization})} \textit{Let $\mathcal{S}$ be a set of systems with an evaluation profile $\mathcal{E}_s$, $\mathcal{A}$ denoting a set of centralization anchors, $\mathcal{U}$ a set of subjects where each subject $u$ is projected over a system topology $T$ of order $t$, with each projection $p \in \mathcal{R}_p$ realized at $n \leq t$ vertices in $T$. An arbitrary subject $u \in \mathcal{U}$ belonging to a system $s \in \mathcal{S}$ with a projection $p$ is said to be \textbf{decentralized} with respect to a given anchor $a \in \mathcal{A}$ if and only if it satisfies the condition $\mu(p_u, a) > 1$.}

\noindent\textbf{Definition 5. (\textit{Decentralization})} \textit{Let $\mathcal{S}$ be a set of systems. For each $s \in \mathcal{S}$, let $\mathcal{U}_s$ be its set of subjects, $\mathcal{A}_s$ its set of centralization anchors, $T_s$ its topology of order $t_s$, and $\mathcal{E}_s \subseteq \mathcal{U}_s \times \mathcal{A}_s$ its evaluation profile. For each $u \in \mathcal{U}_s$, let $p_u$ denote the unique projection of $u$ onto $T_s$, with supporting vertex set $V_{p_u}$ satisfying $1 \leq \lambda(p_u)=|V_{p_u}| \leq t_s$. For any declared pair $(u,a) \in \mathcal{E}_s$, the subject $u$ is said to be \textbf{decentralized} with respect to $a$ if and only if $\mu(p_u, a) > 1$.}

\par

% What can we infer, move this back to logical layer, idea to be included: Distributed-but-centralised architectures, and how can we come to the conclusion that they exist

%\noindent\textbf{Logical Instantiations}
%\par
% blockchain,
% FL,
% DAGs,
% gossip networks,
% hybrid systems? agentic AI?

% Important, since we are using graphs, also part of the modelling
\section{Modeling Layer} % Topological representations 

\noindent\textbf{Topological Representation.} The ontological commitment to graph-theoretic realism entails that any system admitted by $\mathcal{O}$ must have a topology $T$ whose realization is a graph $G = (V, E)$. This is strictly to link the previous layers to analytical layer, such that there exists a mapping from our logical formalizations to computational representations that can be evaluated analytically. 

Crucially, our ontology imposes no constraint on the choice of representational formalism, provided that the chosen structure faithfully encodes the graph $G$. An adjacency matrix $A \in \{0,1\}^{|V| \times |V|}$, an edge list $E \subseteq V \times V$, or any equivalent structure each constitute admissible realizations of $T$, as they all preserve the relational structure required for subject projection. The modeling layer therefore acts as a binding between the abstract topology $T$ defined at the ontological level and its concrete computational instantiation.

\vspace{2mm}

\noindent\textbf{Subject Support and center Regions.} At the ontological level, subject $u$ and topology $T$ are connected through their unique projection particular $p_u$. At the modeling layer, $realizedAt$ derives the non-empty support $V_{p_u}\subseteq V$, its cardinality is $\lambda(p_u)$, while $\delta(p_u)=|\mathcal{R}_{p_u}|$ remains the distinct realization count. We write $V_u$ for $V_{p_u}$ for clearer notation in this paper. Computationally, $V_u$ is a vertex set or binary indicator vector.

An anchor $a$ is modeled by a family $\mathcal{C}_a\subseteq\mathcal{P}(V)$, implementable as a collection of vertex-index sets or a center-vertex incidence matrix. Missing coverage is completed with singleton sets before evaluation. center multiplicity is then the number of rows or sets intersecting $V_u$:
\begin{equation*}
    \mu(p_u,a)=|\{V_c\in\mathcal{C}_a\mid V_c\cap V_u\neq\varnothing\}|.
\end{equation*}
This graph-realistic representation is why an additional center object role is unnecessary.

\vspace{2mm}

\noindent\textbf{Shared Vertices, Subjects, and centers.} A vertex may support several subjects, and a subject may have several realizations at one or more vertices. Once again highlighting that $\delta$ and $\lambda$ but not $\mu$. A multi-vertex support of a subject is distributed, and whether it is decentralized depends on how those vertices intersect $\mathcal{C}_a$. The $1\leq\lambda(p_u)\leq\delta(p_u)$ therefore remains valid, while no inference from $\lambda(p_u)>1$ to $\mu(p_u,a)>1$ is allowed.

\section{Analytical Layer}
%Independent metrics, and measurable outcomes, what claims can we make?
\noindent\textbf{Graph-based Metrics.} Building on the modeling layer, we introduce two analytical operators for a declared pair $(u,a)\in\mathcal{E}_s$. The anchor determines whether the pair passes the logical decentralization boundary through $\mu(p_u,a)>1$; conditional on that boundary, the existing graph operators evaluate the placement and connectivity of $V_u$. Both return values in $[0,1]$.

The first metric, \textit{Void Tolerance} $T_L$, measures vertex-focused decentralization by assessing the resilience of the subject-induced subgraph to vertex removal. Concretely, it asks how many vertices exist whose deletion would isolate a portion of the subgraph, and how large that isolated portion could be. The second operator, \textit{Imperviousness} $I_L$, measures edge-focused decentralization by assessing the minimum number of edge deletions required to compromise the connectivity of any subject-specific vertex, relative to the scale of the topology.

% CHANGE V_U TO PROJECTION COUNT !!!!
\begin{proposition}
Given a system topology $G_t$ and a declared pair $(u,a)$, Void Tolerance is defined as:

% replace |u^I| with \delta(p) 
\begin{equation}
    \begin{aligned}
        T_L(u,a) & = 
        \begin{cases}
            0   & \text{if} \quad \mu(p_u,a) = 1, \\
            1   & \text{if} \quad |G_s| = 0,\, |G_t| > 1,\, \lambda(p_u)>1, \\
            e^{-\frac{r_v^2}{|G_{s}|^{-\epsilon}|G_{t}|}} & \text{otherwise}
        \end{cases} \\[5pt]
        & \text{where} \quad |G_t| > |G_s|, \quad \epsilon > 0
    \end{aligned}
\end{equation}
Where $|G_t|$ is the order of the system topology, $|G_s|$ is the sum order of the isolated subgraphs, $\lambda(p_u)$ is the number of distinct subject-supporting vertices, $r_v$ is the number of vertices whose removal \textit{can} isolate a subgraph, $\mu(p_u,a)$ is the anchor-relative center count, and $\epsilon$ is a subject-anchor weight factor.
% (away from at least a half of corresponding subject specific vertices)

When deciding on the size of the subgraph that is removed (that is, which side is the ``removed'' one), priority is given to the distribution of subject-specific vertices. Following the removal of a vertex, the connected component containing the greatest number of subject-specific vertices is selected as the reference component. All remaining connected components are collectively treated as the isolated subgraph. Thus, if the removed vertex connects three or more components, the component containing the greatest number of subject-specific vertices serves as the reference point, while the vertices of all other components are combined when determining the order of the isolated subgraph. If two or more components contain the same maximal number of subject-specific vertices, the component of greatest order is selected as the reference component. This convention reflects the subject-oriented nature of the metric as the reference component represents the portion of the topology that remains accessible to the greatest concentration of subject realizations.

%any remaining tie may be resolved arbitrarily, as it does not affect the resulting isolated order when the tied components also have equal order.

\end{proposition}
%\newpage
\begin{proposition}
Given a system topology $G_t$ and a declared pair $(u,a)$, Imperviousness is defined as:
    \begin{equation}
    \begin{aligned}
        I_L(u,a) &=
        \begin{cases}
            0 & \text{if} \ |E(G_t)| = 0 \ \text{or} \ \mu(p_u,a)=1,\\
            1 & \text{if} \ r_e=|G_t|-1 \ \text{and} \ \lambda(p_u) > 1,\\
            e^{1-\left( \frac{\lambda(p_u)}{|G_t|} \right)^{-\frac{\epsilon^2}{r_e}}}
            & \text{otherwise}
        \end{cases} \\[5pt]
        & \text{where} \quad |G_t| \ge \lambda(p_u), \quad \epsilon>0
    \end{aligned}
    \end{equation}

Where $\lambda(p_u)$ is the number of distinct vertices supporting the subject projection, $r_e$ is the minimum number of edges whose removal compromises subject functionality, and $\epsilon$ is the subject-anchor weight factor.
\end{proposition}

Both metrics assign $0$ to anchor-relative centralization before evaluating topology-dependent extent. The edge case $T_L(u,a)=I_L(u,a)=0$ arises when $\mu(p_u,a)=1$, including the corrected case in which $\lambda(p_u)>1$ but all supporting vertices lie in one center. The upper bound of $1$ is recovered via limits rather than an explicit construction:
\begin{equation}
    \lim_{r_v\to 0} \left( e^{-\frac{r_v^2}{|G_{s}|^{-\epsilon}|G_{t}|}} \right) = 1, \qquad
    \lim_{r_e\to +\infty} \left( e^{1-\left( \frac{\lambda(p_u)}{|G_t|} \right)^{-\frac{\epsilon^2}{r_e}}} \right) = 1
\end{equation}
This corresponds respectively to the absence of any isolating vertices and a fully connected subject-specific subgraph. In practice, both limits are bounded by the topology: $r_v$ cannot fall below $0$ and $r_e$ is bounded above by $|G_t|-1$. Together, both operators form a vector per declared pair:
\begin{equation}
    \vec{d}_{u,a} = \begin{bmatrix} T_L(u,a) \\ I_L(u,a) \end{bmatrix}
    \quad \text{where} \quad T_L,I_L\in[0,1].
\end{equation}
This representation enables comparison across compatible subject-anchor profiles and is consistent with the profile-relative logical classes.

% Temporary, verify and polish

\vspace{2mm}

\par

\noindent\textbf{Subject-Anchor Weight Factor.} The parameter $\epsilon>0$ encodes the relative importance of a declared pair within the system. A pair regarded as critical may receive a higher value, amplifying structural penalties. No universal assignment procedure is prescribed; the topology may warrant different assignments under different anchor interpretations. In the evaluations below we fix $\epsilon=1$ for every declared pair.

\vspace{2mm}

\noindent\textbf{Multi-dimensional Aggregate.} Vectors for entries in $\mathcal{E}_s$ can be averaged to describe the system's mean analytical degree of decentralization relative to that declared profile. This supports comparative claims only when the same subject-anchor dimensions and center interpretations are used.

This idea naturally develops the multidimensional account for gradable adjectives. On Sassoon's account \cite{sassoon_typology_2013}, a multidimensional adjective is associated with a contextually selected set of dimensions, each represented by a degree function. Later works also propose that these dimensions may be aggregated through universal quantification, existential quantification, or dimension counting, depending on the adjective. Authors of \cite{dambrosio_hedden_2024} preserve this dimension-based representation but generalize the aggregation stage by treating an aggregation function as a function that takes a profile of dimension-specific value functions and returns an overall ordering of the objects under comparison \cite[pp.~259-262, 269-270]{dambrosio_hedden_2024}. 

Simple element-wise averaging instantiates this general framework as an equally weighted aggregation function:
\begin{equation*}
    A(\mathcal{E}_s)=\frac{1}{|\mathcal{E}_s|}
    \sum_{(u,a)\in\mathcal{E}_s}\vec{d}_{u,a}.
\end{equation*}
It makes equal treatment of the included dimensions explicit rather than deriving comparison solely from threshold satisfaction. Aggregate vectors may be compared only after confirming profile compatibility.

Weights within an individual pair's propositions affect that pair's coordinates but do not change how the completed vectors are subsequently treated \textit{en bloc}. Averaging therefore remains equal at the profile-entry level unless a different aggregation rule is explicitly declared.

%[If this isn't sufficient, we can remove subject weight factor and set it to 1 in our equations, the calculations in this paper will remain unaffected]

Simple averaging is appropriate when the coordinate scales and equal-weighting assumption are accepted. Equal weighting remains a modelling choice rather than a context-independent feature of decentralization \cite{dambrosio_hedden_2024}. An aggregate must therefore be reported together with $\mathcal{E}_s$, the relevant $\mathcal{C}_a$ families, parameter assignments, and the aggregation assumption.

% ^^ Talk about aggregating the dimensions in order to describe something as more or less decentralised: https://www.tandfonline.com/doi/full/10.1080/00048402.2023.2277923

\vspace{2mm}

%and $|E(G_t)| = 4$

\noindent\textbf{Analytical Evaluation.} We evaluate both metrics on the logical instantiations and fix $\epsilon=1$. Within each application context, the compared systems use the same evaluation profile and corresponding center interpretations.

\vspace{1mm}

\noindent\emph{Federated Learning.} Encode a star topology with one aggregator, four clients, and $|E(G_t)|=4$ (Figure \ref{fig:centralized-federated-learning}). Use the profile $\mathcal{E}_{FL}$ declared above. For aggregation authority, $V_u=\{v_1\}$ and $\mu(p_{AggAuth},a_{auth})=1$; hence $T_L(u_{AggAuth},a_{auth})=I_L(u_{AggAuth},a_{auth})=0$, matching its logical anchor-centralized classification.

\begin{figure}[h]
    \centering
    \includegraphics[width=0.25\textwidth]{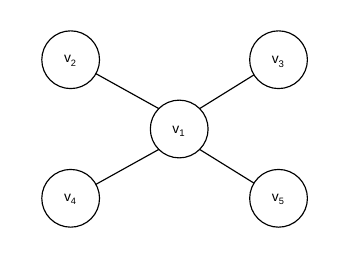}
    \caption{Star-topology centralized federated learning setup with one aggregator and four clients.}
    \label{fig:centralized-federated-learning}
\end{figure}

\noindent For $u_{data}$ and $u_{train}$, $V_u=\{v_2,v_3,v_4,v_5\}$. None of these supporting leaves is an articulation vertex, so $T_L(u_{data},a_{own})=T_L(u_{train},a_{auth})=1$. Each client has one edge to the aggregator, giving $r_e=1$ and $\lambda(p_u)/|G_t|=4/5$; hence:

\begin{equation*}
    I_L(u_{data},a_{own}) = I_L(u_{train},a_{auth})
    = e^{1-(0.8)^{-1}} \approx 0.779
\end{equation*}

From the above, the resulting vectors are:

\begin{equation*}
    \begin{array}{cc}
        \vec d_{data,own}=\vec d_{train,auth}
        = \begin{bmatrix} 1 \\ 0.779 \end{bmatrix}, \\[15pt]
        \vec d_{AggAuth,auth}=\begin{bmatrix}0\\0\end{bmatrix}
    \end{array}
\end{equation*}

As well as the aggregate:
\begin{equation*}
    A(\mathcal{E}_{FL}) = \begin{bmatrix} 0.67 \\ 0.519\end{bmatrix}
\end{equation*}

\vspace{4mm}

For the second example, all five entities hold data and train locally, three distinct authority centers aggregate the model, and $|E(G_t)|=6$ (Figure \ref{fig:decentralized-federated-learning}). For $u_{AggAuth}$, $V_u=\{v_1,v_2,v_3\}$, $\lambda(p_{AggAuth})=3$, and $\mu(p_{AggAuth},a_{auth})=3$. Vertex $v_1$ is the only supporting vertex whose deletion isolates a subgraph of maximum order two, so:

\begin{equation*}
    T_L(u_{AggAuth},a_{auth})
    = e^{-\frac{1^2}{2^{-1}(5)}} \approx 0.67
\end{equation*}

For $I_L(u_{AggAuth},a_{auth})$, at least two edges must be removed to disconnect a supporting vertex, so $r_e=2$. Given $|V_u|=3$, we compute:

\begin{equation*}
    I_L(u_{AggAuth},a_{auth})
    =e^{1-(\frac{3}{5})^{-\frac{1}{2}}}\approx0.748
\end{equation*}

%we target the *weakest* links, out of principle maybe (to argue it) and not the "whole" system, hence the sensitivity.

\noindent For $u_{data}$ and $u_{train}$, every node holds data and trains a local model. As for aggregation authority, one supporting vertex can isolate a subgraph of the same size, so all three declared pairs have the same $T_L$ value.

\begin{figure}[h]
    \centering
    \includegraphics[width=0.25\textwidth]{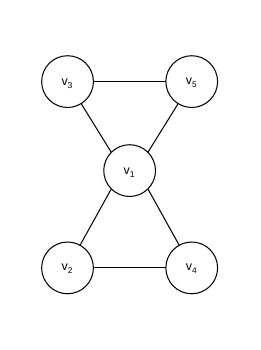}
    \caption{Decentralized federated learning setup with three aggregation authorities, each also acting as data holders and local trainers.}
    \label{fig:decentralized-federated-learning}
\end{figure}

For imperviousness, however, both subjects have the same cardinality of $|V_u|$ (realization of both subjects cover the entire topology). As such, we compute:

\begin{equation*}
    I_L(u_{data},a_{own})=I_L(u_{train},a_{auth})
    =e^{1-(\frac{5}{5})^{-\frac{1}{2}}}=1
\end{equation*}

Since $u_{data}$ and $u_{train}$ are realized across the entire topology, their declared-pair vectors are identical. The three vectors are:

\begin{equation*}
    \begin{array}{c}
        \vec d_{data,own}=\vec d_{train,auth}
        =\begin{bmatrix}0.67\\1\end{bmatrix},\\[15pt]
        \vec d_{AggAuth,auth}=\begin{bmatrix}0.67 \\ 0.748\end{bmatrix}
    \end{array}
\end{equation*}

Along with the aggregate:

\begin{equation*}
    A(\mathcal{E}_{FL}) = \begin{bmatrix} 0.67 \\ 0.916 \end{bmatrix}
\end{equation*}

Now, comparing the two examples based on the results inferred from our analytical layer, we notice how the overall extent of decentralization in the second example is greater than that of a first, despite both being commonly known as instances of federated learning. Void Tolerance of data and training nodes is maximum in the first instance as no disappearing nodes from their sets would result in separating a chunk of a system, meanwhile the same subjects in that example have lower, but relatively high, Imperviousness given that all of them are all help up by only one edge, but there is relatively many of them with respect to the total size of the system. This differs from the second example as sets of realized vertices for both data and training now contain nodes that can isolate a subgraph, meaning that the impact of their absence can have greater effects on the system; indicating that, unlike in the first example, both data and train nodes contain a point of failure. The Imperviousness of data and train subjects in the second example is significantly higher than that of first, provided that they are realized across the entire topology, indicating that it is much harder to isolate those subjects away from the rest of the system with respect to how are they connected across the system, showcasing both pros and cons of both examples across different designs for data and training.

For aggregation authority, the first example has one center and therefore the vector $[0,0]^\mathsf{T}$. The second has three declared authority centers and receives a non-zero vector, although its placement still contains a point of failure. Relative to $\mathcal{E}_{FL}$ and those center interpretations, the first system is partially decentralized and the second fully decentralized. If the three aggregators instead belonged to one authority center, the second aggregation pair would also receive $[0,0]^\mathsf{T}$ despite its unchanged topology. The comparison therefore depends on both structural support and the declared anchor interpretation.

Under the shared profile and stated center families, the second example has the larger aggregate vector and may be described as more decentralized with respect to data ownership, training authority, and aggregation authority.

\vspace{2mm}

\noindent\emph{Blockchain.} We use the same topology $G_t$ and profile $\mathcal{E}_{BC}$ for both examples. Here $|G_t|=7$ and $|E(G_t)|=9$; the topology contains the dense core $\{v_2,v_3,v_4,v_5\}$, endpoint $v_1$, and branch $v_6,v_7$. The stated calculations assume singleton effective centers for the participating entities and $\epsilon=1$.

\begin{figure}[h]
\centering
\includegraphics[width=0.5\textwidth]{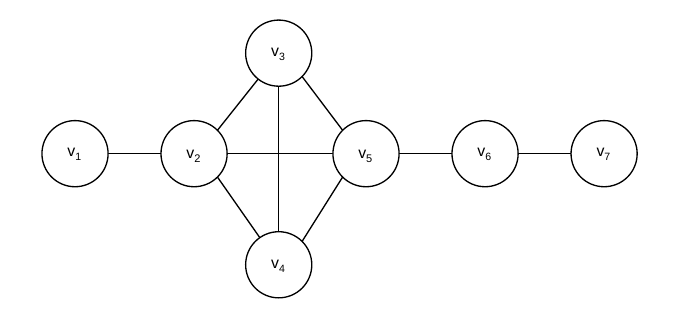}
\caption{Shared blockchain topology used for both blockchain examples.}
\label{fig:blockchain-shared-topology}
\end{figure}

\noindent For the first example, encode a simple permissioned blockchain where one node acts as the block-ordering authority, four nodes store the ledger, and three nodes act as transaction-submission clients. We account for three subjects: Consensus Authority ($u_{\text{cons}}$), Ledger Replication ($u_{\text{ledger}}$), and Transaction Submission ($u_{\text{tx}}$).

For $u_{cons}$, only $v_2$ orders blocks, so $V_u=\{v_2\}$ and $\mu(p_{cons},a_{auth})=1$. By the anchor-centralized edge case:

\begin{equation*}
T_L(u_{cons},a_{auth})=0,\qquad I_L(u_{cons},a_{auth})=0
\end{equation*}

Thus consensus is analytically centralized with respect to authority.

\noindent For $u_{ledger}$, the four core nodes replicate the ledger: $V_u=\{v_2,v_3,v_4,v_5\}$. Removing $v_2$ isolates $v_1$, while removing $v_5$ isolates $\{v_6,v_7\}$. Hence $r_v=2$ and the maximum isolated subgraph has order $|G_s|=2$. We compute:

\begin{equation*}
T_L(u_{ledger},a_{own})=e^{-\frac{2^2}{2^{-1}(7)}}\approx0.319
\end{equation*}

For Imperviousness, the weakest ledger node in the core requires removing three edges to isolate it from the rest of the topology, so $r_e=3$. Since $\lambda(p_{\text{ledger}})=4$, we compute:

\begin{equation*}
I_L(u_{ledger},a_{own})=e^{1-\left(\frac{4}{7}\right)^{-\frac{1}{3}}}\approx0.815
\end{equation*}

\noindent For $u_{\text{tx}}$, suppose that transactions are submitted only by client-facing nodes $V_u(\text{tx})=\{v_1,v_6,v_7\}$. Only $v_6$ can isolate a subgraph when removed, since removing $v_6$ separates $v_7$ from the rest of the system. Hence, $r_v = 1$ and $|G_s| = 1$. We compute:

\begin{equation*}
T_L(u_{tx},a_{auth})=e^{-\frac{1^2}{1^{-1}(7)}}\approx0.867
\end{equation*}

For Imperviousness, either $v_1$ or $v_7$ can be isolated by removing one edge, so $r_e=1$. Since $\lambda(p_{\text{tx}})=3$, we compute:

\begin{equation*}
I_L(u_{tx},a_{auth})=e^{1-\left(\frac{3}{7}\right)^{-1}}\approx0.264
\end{equation*}

\noindent The declared-pair vectors for the first blockchain example are therefore:

\begin{equation*}
    \begin{aligned}
        \vec d_{cons,auth} & =\begin{bmatrix}0\\0\end{bmatrix},\\
        \quad\vec d_{ledger,own} & =\begin{bmatrix}0.319\\0.815\end{bmatrix},\\
        \quad\vec d_{tx,auth} & =\begin{bmatrix}0.867\\0.264\end{bmatrix}
    \end{aligned}
\end{equation*}

With an aggregate:

\begin{equation*}
    A(\mathcal{E}_{BC})=\begin{bmatrix}0.395\\0.36\end{bmatrix}
\end{equation*}

For the second example, keep the exact same system topology, but assign blockchain functionality differently. Encode a blockchain where consensus is performed by the four core nodes, the ledger is replicated by all nodes, and transaction submission is exposed only through two core gateway nodes. We again account for the same three subjects: Consensus Authority ($u_{\text{cons}}$), Ledger Replication ($u_{\text{ledger}}$), and Transaction Submission ($u_{\text{tx}}$).

For $u_{\text{cons}}$, suppose that the four core nodes jointly validate and order blocks $V_u(\text{cons})=\{v_2,v_3,v_4,v_5\}$. This is the same realization used for ledger replication in the first example. Hence, $r_v = 2$, $|G_s| = 2$, $r_e = 3$, and $\lambda(p_{\text{cons}}) = 4$. We compute:

\begin{equation*}
T_L(u_{cons},a_{auth})=e^{-\frac{2^2}{2^{-1}(7)}}\approx0.319
\end{equation*}

and

\begin{equation*}
I_L(u_{cons},a_{auth})=e^{1-\left(\frac{4}{7}\right)^{-\frac{1}{3}}}\approx0.815
\end{equation*}

\noindent For $u_{\text{ledger}}$, suppose that every node stores a copy of the ledger $V_u(\text{ledger})=\{v_1,\dots,v_7\}$. Now the subject is realized across the whole topology. However, $v_2$, $v_5$, and $v_6$ can each isolate some part of the topology when removed. The largest isolated subgraph is still the branch $\{v_6, v_7\}$, so $r_v = 3$ and $|G_s| = 2$. We compute:

\begin{equation*}
T_L(u_{ledger},a_{own})=e^{-\frac{3^2}{2^{-1}(7)}}\approx0.076
\end{equation*}

For Imperviousness, since the ledger subject is realized across the entire topology, $\lambda(p_{\text{ledger}})/|G_t| = 1$. Therefore:
\begin{equation*}
I_L(u_{ledger},a_{own})=e^{1-\left(\frac{7}{7}\right)^{-\frac{1}{r_e}}}=1
\end{equation*}

\noindent For $u_{\text{tx}}$, suppose that transaction submission is exposed only through two non-articulation gateway nodes in the core $V_u(\text{tx})=\{v_3,v_4\}$. Neither $v_3$ nor $v_4$ can isolate a subgraph when removed. Hence, $|G_s| = 0$, $|G_t|>1$, and $\lambda(p_{\text{tx}}) > 1$, so:

\begin{equation*}
T_L(u_{tx},a_{auth})=1
\end{equation*}

For Imperviousness, isolating either $v_3$ or $v_4$ requires removing three edges, so $r_e = 3$. Since $\lambda(p_{\text{tx}})  =2$, we compute:

\begin{equation*}
I_L(u_{tx},a_{auth})=e^{1-\left(\frac{2}{7}\right)^{-\frac{1}{3}}}\approx0.596
\end{equation*}

\noindent The declared-pair vectors for the second blockchain example are therefore:

\begin{equation*}
    \begin{aligned}
        \vec d_{cons,auth} & =\begin{bmatrix}0.319\\0.815\end{bmatrix},\\
        \quad\vec d_{ledger,own} & =\begin{bmatrix}0.076\\1\end{bmatrix},\\
        \quad\vec d_{tx,auth} & =\begin{bmatrix}1\\0.596\end{bmatrix}
    \end{aligned}
\end{equation*}

Alongside the aggregate:

\begin{equation*}
    A(\mathcal{E}_{BC})=\begin{bmatrix}0.465\\0.804\end{bmatrix}
\end{equation*}

Comparing the blockchain examples over the same topology and profile, the analytical layer distinguishes both the declared center interpretation and the placement of each subject. In the first example, consensus is authority-centralized because one authority center performs ordering. Ledger ownership and transaction authority are anchor-decentralized, but their vectors expose different connectivity weaknesses.

In the second example, the same topology supports a different profile of vectors. Under the singleton authority interpretation, consensus reaches four centers; ledger replication reaches every ownership center; and transaction submission reaches two authority centers. The topology-dependent coordinates then expose articulation and isolation effects. If all four consensus vertices were placed inside one authority region, the consensus vector would instead be forced to zero without changing $G$, $\delta$, or $\lambda$.

Under the shared profile and stated center families, the second blockchain example has the larger aggregate vector. That comparison is conditional on these declared anchors and cannot be inferred from topology alone. The above calculations are briefly summarized on Table \ref{tab:analyticalLayer-summary}.

\begin{table*}[ht]
    \centering
    \small
    \setlength{\tabcolsep}{5pt}
    \renewcommand{\arraystretch}{0.95}
    \caption{Summary of analytical decentralization values across the evaluated system instances. $T_L$ denotes Void Tolerance and $I_L$ denotes Imperviousness. Aggregate rows report the element-wise mean over the corresponding declared profile.}
    \label{tab:analyticalLayer-summary}
    \begin{tabularx}{\textwidth}{@{}llccX@{}}
        \toprule
        \textbf{Instance} & \textbf{Declared pair} & $\mathbf{T_L}$ & $\mathbf{I_L}$ & \textbf{Structural summary} \\
        \midrule
        FL-1 & $(u_{\text{data}},a_{\text{own}})$ & 1.000 & 0.779 & No supporting articulation vertex; single-edge client connectivity. \\
             & $(u_{\text{train}},a_{\text{auth}})$ & 1.000 & 0.779 & Same realization and connectivity as data ownership. \\
             & $(u_{\text{AggAuth}},a_{\text{auth}})$ & 0 & 0 & Single authority center; anchor-centralized. \\
             & $A(\mathcal{E}_{FL})$ & 0.670 & 0.519 & -- \\
        \midrule
        FL-2 & $(u_{\text{data}},a_{\text{own}})$ & 0.670 & 1.000 & Topology-wide realization; supporting articulation point present. \\
             & $(u_{\text{train}},a_{\text{auth}})$ & 0.670 & 1.000 & Same realization as data ownership. \\
             & $(u_{\text{AggAuth}},a_{\text{auth}})$ & 0.670 & 0.748 & Three authority centers; one supporting articulation vertex. \\
             & $A(\mathcal{E}_{FL})$ & 0.670 & 0.916 & --. \\
        \midrule
        BC-1 & $(u_{\text{cons}},a_{\text{auth}})$ & 0 & 0 & Single ordering authority; anchor-centralized. \\
             & $(u_{\text{ledger}},a_{\text{own}})$ & 0.319 & 0.815 & Core replication; two supporting articulation vertices. \\
             & $(u_{\text{tx}},a_{\text{auth}})$ & 0.867 & 0.264 & Limited articulation exposure; weak edge-isolation resistance. \\
             & $A(\mathcal{E}_{BC})$ & 0.395 & 0.360 & -- \\
        \midrule
        BC-2 & $(u_{\text{cons}},a_{\text{auth}})$ & 0.319 & 0.815 & Four-node consensus realization in the dense core. \\
             & $(u_{\text{ledger}},a_{\text{own}})$ & 0.076 & 1.000 & Topology-wide replication; several articulation vertices. \\
             & $(u_{\text{tx}},a_{\text{auth}})$ & 1.000 & 0.596 & Two non-articulation core gateways; stronger edge connectivity. \\
             & $A(\mathcal{E}_{BC})$ & 0.465 & 0.804 & -- \\
        \bottomrule
    \end{tabularx}
\end{table*}

\section{Implementation}
\label{sec:implementation}

We implemented the ontology as a browser-only research sandbox that can be hosted as a static GitHub Pages site \cite{implementation}. The interface is separated into semantic HTML, modular style sheets, and JavaScript modules for the evaluator, graph workspace, simulator, result views, and application coordination. No server-side component or build step is required, so an analysis remains local to the browser and the complete project can be reproduced from the repository.

The workspace represents a system as a finite simple undirected graph and stores each declared evaluation dimension as a subject-anchor pair $(u,a)\in E_s$. Anchors are first-class system inputs: the user declares a contextual interpretation and defines its single center family $C_a$ by painting named vertex regions directly on the canvas. Regions may overlap, and uncovered vertices are completed as singleton centers. Profile entries then select an anchor and record $\delta(p_u)$, the supporting vertices used to derive $\lambda(p_u)$, and $\epsilon_{u,a}$. Consequently, every entry using $a$ reuses the same $C_a$ when computing $\mu(p_u,a)$, while a distributed subject may still be centralized under one shared center.

For analytical evaluation, one iterative low-link traversal derives the articulation structure of the full topology in $O(|V|+|E|)$ time and is reused across all profile entries. The evaluator calculates $r_v$, $|G_s|$, and $r_e$, applies the anchor-centralized zero cases, and returns $\vec d_{u,a}=[T_L,I_L]^\mathsf{T}$ and the profile mean $A(E_s)$. Seeded simulations separately randomize topology, support, multiplicity, contextual anchor assignment, and one shared center family per generated anchor. The supplied ownership, authority, trust, and governance generators use transparent actor/domain heuristics (including single operators, organizational partitions, hub-led domains, and overlapping consortia). Automated tests reproduce the paper's federated-learning and blockchain examples and verify shared, overlapping, and singleton-completed center families.

\section{Discussion} % Here is where our comparison will go
% The question we want to protect ourselves from is: whether the ontology clarifies distinctions that existing frameworks/definitions confuse. We do that by showing how other methods fail to do what we can (also making our approach not trivial)

Having established the ontology, its formal properties, and analytical capabilities, we now reflect on what the proposed definition contributes relative to the existing literature and the broader implications that follow from it.

\subsection{Distribution and Decentralization as Distinct Properties.} 
A central contribution of the framework is the formal separation. The value $\delta(p_u)$ counts realization particulars, $\lambda(p_u)$ counts supporting vertices, and $\mu(p_u,a)$ counts center regions reached under a declared anchor. Each with their separate conditions in regards to which aspects of decentralization are they meant to address, formalizing distribution as a placement condition $\lambda(p_u)>1$ and decentralization as an anchor-relative condition $\mu(p_u,a)>1$. Simultaneously encapsulating various circumstances in which decentralization may exist, including where several vertices may lie in one authority or ownership region, yielding $\lambda(p_u)>1$ and $\mu(p_u,a)=1$. This makes the distinction acknowledged in prior literature \cite{rossi_towards_2019,lui_2026_blockchainbased} enforceable in the formalism.

\subsection{Epistemological Distinction between Distribution and Decentralization.} 
Distribution can be observed from the selected topology, whereas anchor-relative decentralization additionally requires an interpretation of what constitutes a center of centralization (which led our intuition that for any proper formalization of this concept needs to be cross-contextually operational). The fact that a function is implemented at several vertices does not establish that those vertices have independent ownership, authority, or trust. The framework therefore makes a practitioner state $a$ and $\mathcal{C}_a$ rather than allowing a deployment count to stand in for a stronger organizational claim. If no shared or external center is declared, singleton completion makes that modelling assumption explicit rather than silently presuming it.

\subsection{Relational Character of Subject Multiplicity.} 
Although all three functions return their own independent values, their extensions are relationally derived. Our projection particular connects the system, topology, and subject through \textit{hasProjection}, \textit{ofSubject}, and \textit{inTopology}. Whereas $hasRealization$ supplies $\mathcal{R}_{p_u}$ and $realizedAt$ supplies $V_{p_u}$. The anchor relation supplies $a$, while $\mathcal{C}_a$ gives its extensional center interpretation. Successfully separating these fundamental concepts to decentralization all the while sustaining their intertwined nature. As such, the properties from which these functions are derived are relational \cite{guarino1998,rossi_towards_2019}.

Graph realism supports the simplification because each center is a vertex subset rather than a second parallel object hierarchy. The predicate inside $\mu$ intuitively asks whether some realization of $p_u$ is located at some vertex inside that subset. The graph's edges then determine the resilience of the same support in the analytical layer. Anchor-relative center multiplicity and edge sensitivity are therefore separate, allowing $\mu$ to classify which center regions are reached, while $T_L$ and $I_L$ evaluate how the supporting vertices are connected.

\subsection{Collective Amnesia Regarding the Relational Aspect of Decentralization.} 
The broader literature exhibits what might be described as collective amnesia about the relational and interpretative aspects of decentralization. Foundational treatments attended to network relations \cite{baran_distributed_1964,lamport_byzantine_1982,brewer_towards_2000}, while many later measures focus on participant or resource distributions \cite{nakamoto_bitcoin_2008,di_bona_concept_2023,srinivasan_quantifying_2017,lin_measuring_2021}. Participant count does not determine whether participants fall under one or several effective centers, and neither a count nor $\mu$ alone captures the graph vulnerabilities evaluated by Void Tolerance and Imperviousness. 

\subsection{Analytical Depth Beyond Existing Metrics.} 
The analytical layer adds information after the logical anchor boundary has been established with an evaluation profile based on a given topology. Gini coefficients, entropy, and threshold coefficients summarize different scalar properties and do not by themselves encode center regions or the fragility of subject-supporting subgraphs. Void Tolerance and Imperviousness instead assess vertex and edge-removal effects for each declared pair. The resulting $\vec d_{u,a} = [T_L(u,a), I_L(u,a)]^\mathsf{T}$ exposes two structural axes, while $\mu(p_u,a) = 1$ prevents a distributed placement under one center from receiving a non-zero decentralization vector. 

\subsection{Visualization and Cross-Contextual Comparison.} 
By placing $\vec d_{u,a}$ in $[0,1]^2$, compatible subject-anchor pairs can be compared along the same two axes. A vector approaching $[1,1]^\mathsf{T}$ simply indicates high values under both graph operators, on the other hand, $[0,0]^\mathsf{T}$ is forced for an anchor-centralized pair. Intermediate points reveal whether the weaker coordinate is vertex or edge-focused allowing for cross-context comparisons only when the compared profiles, anchor meanings, center-construction rules, topology abstraction, and parameters are disclosed.

\subsection{Bringing Formal Reasoning closer.} 
The ontology bridges informal judgement and formal reasoning by exposing, rather than hiding, the choices that precede inference. Analysts identify subjects, anchors, center regions, and the evaluation profile according to the intended function of the system. Formal layer then computes $\delta$, $\lambda$, and $\mu$ and applies fixed truth conditions. This reflects the role of ontological commitment described by Guarino \cite{guarino1998}, where context supplies an interpretation, while the ontology constrains the consequences (which in our case achieves contextual operationalism). The result is not interpretation-free, but it is auditable because a claim such as ``data is decentralized with respect to ownership'' can be checked against an explicit $p_u$, $a$, and $\mathcal{C}_a$.

\section{Limitations}

Although the proposed ontology establishes a formally grounded and cross-contextual account of decentralization, it remains subject to several limitations that arise from the assumptions necessary for its construction. These limitations should be understood as consequences of the intended scope of the framework rather than deficiencies of its formal foundations.

\subsection{Classes of Decentralized Systems}

The present ontology is intentionally concerned with decentralization as a structural property of communication systems. As a consequence, we do not attempt to capture behavioral, economic, or performance-related properties that are frequently associated with decentralized systems, such as latency, incentive compatibility, robustness, or fairness. While these characteristics undeniably influence practical system evaluation, they constitute distinct analytical dimensions whose formal treatment requires additional ontological commitments beyond those developed here.

Likewise, the ontology evaluates communication systems as snapshots rather than temporally evolving entities. Topology evolution or changing ownership and authority must be represented through successive graphs and center families rather than explicit reasoning. Anchor-relative classification is therefore indexed implicitly to the time at which $G$ and $\mathcal{C}_a$ are asserted.

\subsection{Graph Realism as a Basis for System Representation}

Our framework adopts graph representations as ontologically meaningful abstractions of communication systems. This commitment follows from the ontological realist position developed throughout the paper, namely that relations between entities are treated as genuine features of the systems under investigation rather than merely convenient modeling artifacts. Nevertheless, this commitment necessarily limits the scope of the ontology to systems that admit faithful graph representations.

Certain interactions are more naturally expressed through higher-order or non-pairwise structures beyond ordinary graphs. Anchor center families can distinguish some semantically different systems with identical communication topologies, but only when those distinctions are supplied as vertex regions. They do not make an ordinary graph capable of representing every higher-order dependency. Ontological realism therefore remains constrained by both the topology abstraction and the fidelity of the anchor interpretation.

\subsection{Ontological Realism, Nominalism, and the Scope of Structural Decentralization}

A further limitation concerns the metaphysical status of decentralization and the extent to which its associated properties can be represented without altering the ontology's foundational commitments. In the limited sense relevant to our work, a nominalist treatment would use ``decentralized'' as a label without committing to a common structure that determines when it applies. Our ontology instead gives explicit truth conditions over a subject projection and an anchor interpretation ($\mu(p_u,a)=1$ or $\mu(p_u,a)>1$). This follows the realist position that an ontology should constrain its models according to an intended conceptualization \cite{guarino1998, arp_building_ontologies}.

The anchor layer makes two systems distinguishable even when $G$, $\mathcal{R}_{p_u}$, and $V_{p_u}$ are identical. If all supporting vertices are assigned to one ownership region in the first system and to several ownership regions in the second, $\delta$ and $\lambda$ remain equal while $\mu$ differs. Ownership, authority, and trust are therefore qualify the interpretation under which a subject's centers are counted. Our ontology does not discover beneficial ownership, effective authority, collusion, shared credentials, delegation, or hidden common control from the communication graph. A center family is an analyst-supplied extensional abstraction over vertices. If the supplied $\mathcal{C}_a$ is wrong or incomplete, the resulting $\mu$ and every downstream classification may also be wrong. Identifying stable real-world principals and deciding whether nominally separate entities constitute one effective center remain modelling and evidential tasks outside the counting rule.

The singleton completion convention is particularly consequential. Treating every uncovered vertex as an independent effective center ensures total coverage and makes $\mu$ defined, but absence of a declared shared center is not evidence of actual independence. Results obtained through the default family must therefore be reported as conditional on that convention. Where evidence is unavailable, analysts should distinguish ``no shared center declared'' from an empirical claim that no shared center exists.

The extensional simplification also omits identity, intensity, and internal organization. Because $\mathcal{C}_a$ is a set of vertex subsets, two proposed centers with exactly the same vertex extension collapse into one effective center from the view of our ontology. The model therefore assumes that centers are distinguished by their graph-realistic regions rather than by a separate identity. Every intersected region contributes one to $\mu$ regardless of its size, voting weight, threshold, or degree of control. If overlapping regions are used, one vertex may contribute to more than one counted center and that interpretation must be justified. Richer principal, delegation, weighted-control, or temporal models could refine these cases, but they would add commitments beyond the corrective anchor/center abstraction developed here.

Accordingly, the framework establishes graph-realistic, subject, and anchor-relative decentralization under declared center semantics. It does not assert organizational independence in every broader social, legal, economic, or adversarial sense. Its transferable element is the invariant counting rule; the truth of a concrete result remains conditional on the fidelity of $G$, $p_u$, $a$, $\mathcal{C}_a$, and $\mathcal{E}_s$.

\subsection{Ad Hoc Metric Construction and Interpretative Scope}

Void Tolerance and Imperviousness are purpose-built operators whose functional forms are not uniquely entailed by the axioms. Exponential normalization, $\epsilon$, and aggregation of subject-anchor vectors are analytical design choices (even if rooted in formal semantics), and alternative functions could use the same ontology while producing different scales or orderings.

In our view, \emph{ad hoc} means constructed for a particular analytical purpose, not arbitrary. Void Tolerance concerns subject-relevant vertex removal, meanwhile, Imperviousness concerns the edge-removal effort needed to isolate a supporting vertex. The anchor affects the logical boundary, but the formulas above that boundary still operate on $V_{p_u}$ rather than on the internal organization of each center. Consequently, two non-centralized anchor interpretations with the same support may receive identical vectors even when their center counts differ.

We underline this distinction because a formally specified statistic does not, by itself, define decentralization. The Gini coefficient measures inequality within a distribution without establishing that equality is identical to decentralization \cite{lin_measuring_2021}. A singleton population, for example, has no internal inequality. Shannon's Entropy measures uncertainty, while the Herfindahl-Hirschman Index measures concentration. Treating any of these as decentralization requires an additional premise about the subject, anchor, and relevant centers.

% In much of the existing literature, this interpretative premise is established informally. A resource distribution, validator-weight distribution, or participant population is selected, a mathematically well-defined statistic is applied to it, and the resulting value is subsequently described as a degree of decentralization.

Based on the above, the formal precision of the statistic does not automatically transfer to the conceptual inference that the statistic measures decentralization itself. Different metrics may therefore produce equally valid descriptions of inequality, concentration, uncertainty, or attack thresholds while representing different and potentially incompatible interpretations of decentralization.

Nevertheless, the purpose-built nature of Void Tolerance and Imperviousness remains a genuine methodological limitation of our approach. Their normalization, sensitivity to $\epsilon$, and aggregation procedure require independent analysis and comparison with alternative graph and concentration measures. Numerical results must be reported with the topology, evaluation profile, center families, parameter assignments, and aggregation assumptions. Their values are conditional analytical outputs, not context-independent constants of decentralization. 

The analytical contribution therefore does not rest on claiming that the two metrics are uniquely correct. It makes the relationship between definitions, graph properties, metric construction, and their interpretation explicit. What's worth noting is that our operators remain replaceable without undermining the foundational separation established by $\delta$, $\lambda$, and $\mu$.

\subsection{Formal and Informal Reasoning}

As argued throughout this work, following Finocchiaro \cite{Finocchiaro_2005, Finocchiaro2005-FINAAA-4}, we treat formal specification and context-sensitive judgment as complementary rather than mutually exclusive. Accordingly, we constrain informal judgment by making its conceptual commitments explicit within a formally specified structure.

In particular, identifying subjects, selecting anchors, constructing center regions, and abstracting real systems remain dependent on informed human judgment and evidence. Different analysts may construct different profiles while remaining formally consistent. The framework reduces ambiguity by requiring those choices to be exposed, but it cannot eliminate them. As such, reproducibility requires publishing $\mathcal{E}_s$ and every relevant $\mathcal{C}_a$ alongside the graph.

\section{Conclusions}

This work addressed the Decentralization Problem by introducing a formal, graph-based ontology that defines decentralization independently of application domain and separates three fundamental questions: how many realization particulars exist ($\delta$), across how many vertices they are placed ($\lambda$), and how many center regions they reach under an anchor such as ownership or authority ($\mu$). By incorporating ontological commitments, formal semantics, graph-theoretic modeling, and subject-specific analytical metrics, the proposed framework establishes a transferable basis for reasoning about decentralization across computer communication systems. In doing so, it formally distinguishes decentralization from distribution and demonstrates that structurally identical topologies may exhibit different decentralization properties depending on the subject under consideration.

The ontology was validated through two types of communication systems, showing that it provides consistent classifications while exposing structural characteristics that conventional scalar measures fail to capture. More broadly, this work establishes decentralization as a formally definable relational property rather than an application-dependent or solely informally interpreted concept, providing a common foundation for future theoretical and practical work on decentralized computer systems.

\section{AI Disclosure Statement}

Artificial intelligence tools were used primarily for grammar, language checks, and for stress-testing of the manuscript. All outputs were critically reviewed by the authors and revised accordingly. Authors retain full responsibility for the article.

\bibliographystyle{IEEEtran}
\bibliography{references}

\appendix

% We could potentially shove this into the appendix
\noindent\textbf{A1. Aristotelian Definitions}
Having created our taxonomy, we establish Aristotelian definitions (as per recommended best practices \cite{arp_building_ontologies}) for all terms that can be derived using our ontology:\par
\vspace{1mm}

\noindent \textsc{Anchor} = \textit{def.} A system-relative interpretation, such as ownership or authority, under which effective centers are distinguished.\par
\noindent \textsc{Center Region} = \textit{def.} A member $V_c\in\mathcal{C}_a\subseteq\mathcal{P}(V_T)$ representing one effective center under anchor $a$.\par
\noindent \textsc{Centralized Subject} = \textit{def.} A declared pair $(u,a)$ whose unique projection particular satisfies $\mu(p_u,a)=1$.\par
\noindent \textsc{Decentralized Subject} = \textit{def.} A declared pair $(u,a)$ whose unique projection particular satisfies $\mu(p_u,a)>1$.\par

\noindent \textsc{Centralized System} = \textit{def.} A system for which every entry in the declared evaluation profile is anchor-centralized.\par
\noindent \textsc{Partially Decentralized System} = \textit{def.} A system whose declared evaluation profile contains at least one anchor-centralized and at least one anchor-decentralized entry.\par
\noindent \textsc{Fully Decentralized System} = \textit{def.} A system with a non-empty declared evaluation profile every entry of which is anchor-decentralized.\par
\noindent \textsc{Undistributed Subject} = \textit{def.} A subject whose projection support satisfies $\lambda(p_u)=1$.\par
\noindent \textsc{Distributed Subject} = \textit{def.} A subject whose projection support satisfies $\lambda(p_u)>1$.\par
\noindent \textsc{Undistributed System} = \textit{def.} A system where all subjects of which are, altogether, held on a singular node on the associated topology. \par
\noindent \textsc{Distributed System} = \textit{def.} A system all subjects of which are, altogether, projected across more than one node on the associated topology.\par

\end{document}